\newif\ifConference
\newif\ifJournal
\newif\ifAnonymous
\newif\ifAppendix
\newif\ifFinal

\Conferencetrue
\Appendixtrue
\Finaltrue

\documentclass[a4paper,USenglish,cleveref,autoref,thm-restate,numberwithinsect,nameinlink]{lipics-v2021}
\RequirePackage{etex}

\usepackage{etex}
\usepackage[utf8]{inputenc}
\usepackage[T1]{fontenc}

\IfPackageLoadedTF{babel}{}{
  \usepackage[english]{babel}
}

\usepackage[charter]{mathdesign}
\usepackage{eucal}
\usepackage{microtype}

\usepackage{booktabs}

\IfPackageLoadedTF{hyperref}{}{
  \usepackage[bookmarks=false, pdftex, colorlinks=true, allcolors=blue]{hyperref}
}
\hypersetup{
	colorlinks,
	linkcolor={red!50!black},
	citecolor={blue!90!black!80},
	urlcolor={blue!80!black}
}

\IfPackageLoadedTF{cleveref}{}{
\usepackage[capitalize, nameinlink]{cleveref}
}

\usepackage{amsfonts,amsmath,amsthm}
\usepackage{graphicx}
\usepackage[vlined,ruled,linesnumbered]{algorithm2e}
\usepackage{mathtools}
\usepackage{xcolor}
\usepackage{nicefrac}
\usepackage{tabularx}
\usepackage{xspace}

\ifFinal
\usepackage[disable]{todonotes}
\else
\usepackage[notref, notcite]{showkeys}
\usepackage{todonotes}
\fi

\usepackage{etoolbox}
\newcommand{\appendixproofs}{}
\newcommand{\toappendix}[1]{\gappto{\appendixproofs}{#1}}
\toappendix{
	
	\section{Appendix}
	
	\setcounter{section}{0}
}

\newcommand{\thmtoappendix}[3]{\toappendix{
		\section{Proof of~\Cref{#1}}
		\renewcommand\thetheorem{\ref{#1}}
		#2
		
		#3
	}
}

\usepackage{tikz}
\usetikzlibrary{arrows.meta,shapes,positioning,calc,decorations.pathreplacing,decorations.markings,decorations.pathmorphing,patterns}

\pgfdeclarelayer{background2}
\pgfdeclarelayer{background}
\pgfdeclarelayer{foreground}
\pgfsetlayers{background2,background,main,foreground}

\usepackage[edges]{forest}

\tikzstyle{bold}=[draw, line width=2pt]
\tikzstyle{optional}=[dashed]
\tikzstyle{path}=[decorate, decoration={snake, amplitude=.6mm}]

\tikzstyle{small}=[inner sep=2pt]
\tikzstyle{tiny}=[inner sep=1.7pt]
\tikzstyle{textnode}=[inner sep=0pt]

\tikzstyle{triangle}=[draw, regular polygon, regular polygon sides=3]
\tikzstyle{vertex}=[circle, draw, fill=white]
\tikzstyle{reti}=[vertex, fill=black]
\tikzstyle{leaf}=[vertex, rectangle]
\tikzstyle{leaf2}=[vertex, regular polygon, regular polygon sides=3]

\tikzstyle{smallvertex}=[vertex, small]
\tikzstyle{smallleaf}=[leaf, inner sep=3.3pt]
\tikzstyle{smallleaf2}=[leaf2, inner sep=1.7pt]
\tikzstyle{smalltriangle}=[triangle, inner sep=1.5pt]
\tikzstyle{smallreti}=[reti, small]

\tikzstyle{match}=[edge,line width=3pt]
\tikzstyle{edge}=[draw,-]
\tikzstyle{arc}=[draw,arrows={-Latex[length=6pt]}]
\tikzstyle{boldarc}=[draw, bold, arrows={-Latex[length=10pt]}]
\tikzstyle{revarc}=[draw, arrows={Latex[length=6pt]-}]
\tikzstyle{boldrevarc}=[draw, bold, arrows={Latex[length=10pt]-}]

\newcommand{\nextnode}[5][vertex]{\node[small#1] (#2) at ($(#3)+(#4)$) {} edge[#5] (#3);}
\newcommand{\nextnodelab}[7][vertex]{\node[small#1] (#2) at ($(#3)+(#4)$) {}; \draw[#5] (#2) -- (#3) node[pos=#6] {#7};}

\newcommand{\edgelab}[3][0]{\rotatebox{#1}{\rotatebox{-#1}{#2} \rotatebox{-#1}{#3}}}

\makeatletter
\@ifundefined{claim}{{\upshape\itshape}{\upshape\rmfamily}}{}
\@ifundefined{construction}{\newtheorem{construction}{Construction}{\upshape\itshape}{\upshape\rmfamily}}{}
\@ifundefined{rrule}{{\upshape\itshape}{\upshape\rmfamily}}{}
\@ifundefined{observation}{\newtheorem{observation}{Observation}{\upshape\itshape}{\upshape\rmfamily}}{}
\makeatother

\newcommand{\kommentar}[1]{}
\newcommand{\Oh}{\ensuremath{O}}

\newcommand{\proofpara}[1]{\smallskip
	
	\noindent\textit{#1.}}

\DeclareMathOperator{\swwithoutN}{{sw}}
\DeclareMathOperator{\nswwithoutN}{{nsw}}

\DeclareMathOperator{\parents}{parents}

\DeclareMathOperator{\off}{off}

\DeclareMathOperator{\DP}{DP}

\newcommand{\w}{{\ensuremath{\omega}}\xspace}

\newcommand{\yes}{{\normalfont\texttt{yes}}\xspace}

\newcommand{\NP}{{\normalfont{NP}}\xspace}

\newcommand{\Instance}{{\ensuremath{\mathcal{I}}}\xspace}

\newcommand{\sw}[1][N]{{\ensuremath{\swwithoutN_{#1}}}\xspace}
\newcommand{\nsw}[1][N]{{\ensuremath{\nswwithoutN_{#1}}}\xspace}
\newcommand{\GW}[1][\Gamma]{{\ensuremath{GW_{#1}}}}

\newcommand{\NN}{\mathbb{N}}

\newlength{\boxindent}
\newcommand{\problemdef}[3]{
	\begin{trivlist}
		\item \tikz{\node [draw,inner sep=4.5pt] {\begin{minipage}{\textwidth-9.4pt}
					\raggedright
					\normalsize\textsc{#1}
					
					\smallskip
					
					\settowidth{\boxindent}{\textbf{Question:} }%
					\leavevmode\hangindent=\boxindent
					\rlap{\textbf{Input:}}\phantom{\textbf{Question:}}#2
					
					\leavevmode\hangindent=\boxindent
					\textbf{Question:}#3
			\end{minipage}};}
	\end{trivlist}
}

\newcommand{\PROB}[1]{{{\textsc{#1}}}\xspace}

\newcommand{\smashsum}[2][lr]{\ensuremath{\smashoperator[#1]{\sum_{#2}}}}
\newcommand{\smashprod}[2][lr]{\ensuremath{\smashoperator[#1]{\prod_{#2}}}}

\let\citep\cite

\newcommand{\belo}[1][N]{\ensuremath{<_{#1}}}
\newcommand{\beloeq}[1][N]{\ensuremath{\leq_{#1}}}
\newcommand{\abov}[1][N]{\ensuremath{>_{#1}}}
\newcommand{\aboveq}[1][N]{\ensuremath{\geq_{#1}}}

\newcommand{\zeroone}[1]{\ensuremath{\left[#1\right]}}

\newcommand{\APD}[1][N]{{\ensuremath{\text{APD}_{#1}}}\xspace}
\newcommand{\PD}[1][N]{{\ensuremath{\text{PD}_{#1}}}\xspace}
\newcommand{\EPD}[1][N]{{\ensuremath{\mathbb{E}\text{PD}_{#1}}}\xspace}

\newcommand{\APDlong}{average-tree phylogenetic diversity\xspace}
\newcommand{\APDX}{\ensuremath{\APD(X_N)}\xspace}
\newcommand{\MaxAPD}{\PROB{Max-APD}}

\newcommand{\ucNAP}{\PROB{unit-cost-NAP}}

\newcommand{\indeg}[1][]{\ensuremath{{\operatorname{deg}^-_{#1}}}}
\newcommand{\outdeg}[1][]{\ensuremath{{\operatorname{deg}^+_{#1}}}}
\newcommand{\edgeproba}[3][N]{\ensuremath{\varphi_{#1}\left({#2}, {#3}\right)}}
\newcommand{\haspath}[3][N]{\ensuremath{{#2}\geq_{#1}{#3}}}
\newcommand{\nhaspath}[3][N]{\ensuremath{{#2}\not\geq_{#1}{#3}}}
\newcommand{\haspathzo}[3][N]{\zeroone{\haspath[#1]{#2}{#3}}}
\newcommand{\nhaspathzo}[3][N]{\zeroone{\nhaspath[#1]{#2}{#3}}}

\newcommand{\inedges}[1][v]{\ensuremath{E}^{#1-}}
\newcommand{\outedges}[1][v]{\ensuremath{E}^{#1+}}
\newcommand{\childprob}[2][\Delta]{\ensuremath{p^{#2}_{\leq #1}}}

\newcommand{\compsw}{\mathfrak{S}} 

\newcommand{\swis}{\ensuremath{\mathcal{S}}\xspace} 

\newcommand{\ptable}[2]{\ensuremath{\text{p}\left[#1,#2\right]}}
\newcommand{\dtable}[2]{\ensuremath{\text{d}\left[#1,#2\right]}}

\let\emptyset\varnothing

\title{Average-Tree Phylogenetic Diversity Parameterized by Scanwidth and Invisibility}

\ifAnonymous

\author{Anonymous Author(s)}
{Anonymous Affiliation(s)}
{}
{}
{}
\authorrunning{Anonymous Author(s)}

\else

\author{Leo van Iersel}
{Delft University of Technology, Delft, The Netherlands \and \url{LeoVanIersel.nl}}
{l.j.j.vanIersel@tudelft.nl}
{https://orcid.org/0000-0001-7142-4706}
{Partially funded by the Dutch Organisation for Scientific Research (NWO) grant OCENW.KLEIN.125 and OCENW.M.21.306.}
\author{Mark Jones}
{Middlesex University, London, United Kingdom}
{m.jones@mdx.ac.uk}
{https://orcid.org/0000-0002-4091-7089}
{}
\author{Jannik Schestag}
{Delft University of Technology, Delft, The Netherlands \and
\url{www.JannikSchestag.eu}}
{j.schestag@uib.no}
{https://orcid.org/0000-0001-7767-2970}
{Partially funded by the Dutch Research Council (NWO), project OCENW.GROOT.2019.015 “Optimization for and with Machine Learning (OPTIMAL)”.}
\author{Celine Scornavacca}
{ISEM, Université de Montpellier, CNRS, IRD, EPHE, Montpellier, France \and \url{https://sites.google.com/view/celinescornavacca}}
{celine.scornavacca@umontpellier.fr}
{https://orcid.org/0009-0004-0179-9771}
{Partially funded by French Agence Nationale de la Recherche through the CoCoAlSeq Project (ANR-19-CE45-0012).}
\author{Mathias Weller}
{LIGM, CNRS, Univsersité Gustave Eiffel, Marne-la-Vallée, France}
{mathias.weller@tu-berlin.de}
{https://orcid.org/0000-0002-9653-3690}{}

\Copyright{Leo van Iersel, Mark Jones, Jannik Schestag, Celine Scornavacca, and Mathias Weller}

\authorrunning{van Iersel, Jones, Schestag, Scornavacca, and Weller}

\fi 

\date{\today}

\newcommand{\todos}[1]{\todo[color=green!70!black!70]{#1}}
\newcommand{\todosi}[1]{\todo[inline,color=green!70!black!70]{#1}}

\allowdisplaybreaks



\hideLIPIcs
\nolinenumbers

\begin{document}

	\maketitle
	
  \begin{abstract}
    We investigate parameterized algorithms for computing the average-tree phylogenetic diversity (APD) in rooted phylogenetic networks,
    studying the problem under different structural parameters that capture the deviation of a network from a tree.
    Our primary parameter is the scanwidth,
    a measure of the tree-likeness of a given directed acyclic graph.
    We show that a subset of taxa with maximum APD can be found in polynomial time in phylogenetic networks with scanwidth at most~2,
    but becomes NP-hard in networks with scanwidth at most~$3$. 
    Further, we design an algorithm that computes the APD of a given set of taxa in~$\Oh(2^{s} \cdot n)$~time, given a width-$s$ tree-extension, with~$n$ vertices of the input network.
    Finally, we give a linear-time algorithm for computing the APD of a given set of taxa
    if the network induced by these taxa is reticulation-visible.
    We generalize this algorithm to an FPT algorithm with respect to the maximum number of invisible reticulations of any biconnected component of the induced network, assuming the indegrees are bounded by a constant.
  \end{abstract}

\todosi{Submission IPEC:
June 29\\
No page/line limit, but reviewers may read up to 500 lines.}

\newpage

\section{Introduction}
Human activities are driving what is widely recognized as a sixth mass extinction~\citep{Ripple2017warning}. 
The question of how best to preserve biodiversity under limited financial and temporal resources has become increasingly urgent~\citep{weitzman1998noah}. This challenge has motivated the development of quantitative measures that aim to guide conservation priorities. 

Faith proposed phylogenetic diversity as a tractable proxy for feature diversity~\cite{faith1992conservation}, which is
the total amount of evolutionary features represented by a set of taxa---that is, a set of species, populations, or individuals.
To make feature diversity easier to compute, phylogenetic diversity (PD) is often used as a tree-based proxy
that does not require detailed trait data~\citep{faith1992conservation}.
Although this proxy is not always biologically meaningful~\cite{wicke2021formal}, PD remains a widely adopted choice.

\def\maybehide#1#2{\if\relax#1\relax\phantom{#2}\else#2\fi}

\looseness=-1
Formally, when evolutionary relationships of a set of taxa~$X$ 
are represented by a rooted phylogenetic tree~$T$ with leaf set~$X$ and edge lengths,
the PD of a set~$Z \subseteq X$ is defined as
the sum of the lengths of all edges that lie on a path from the root of~$T$ to a leaf in~$Z$~\citep{faith1992conservation}.
In the classical tree setting, PD is well understood from an algorithmic perspective.
Given a phylogenetic tree~$T$ and an integer~$k$, a subset of~$k$ species maximizing PD can be found efficiently using a greedy algorithm~\citep{pardi2005species,steel2005phylogenetic}.
However, reticulate evolutionary events such as hybrid speciation, lateral gene transfer, and recombination---which allow species to inherit genetic material from multiple ancestors---make evolutionary histories more complex than a tree can capture.
In such cases, evolutionary relationships are more accurately modeled by rooted phylogenetic networks rather than trees~\citep{huson2010phylogenetic}.
Rooted phylogenetic networks are directed acyclic graphs with a single source vertex (the ``root'')
that describe the evolutionary history of a set~$X$ of taxa,
represented by the leaves (sinks) of the network.
Such networks, when equipped with edge lengths and inheritance probabilities, can be used to estimate the biodiversity of a subset of the taxa.  
Edge lengths can have two different meanings: either they correspond to the expected number of substitutions per site or to elapsed evolutionary time (for a class of trees and networks called \emph{ultrametric}). In both cases, longer edges mean a more distant (looser) evolutionary relationship between the connected vertices.
Inheritance probabilities tell us how much of the genetic material of an organism that results from a reticulate evolutionary event comes from each parent when lineages merge.

\begin{figure}[t]
    \centering
    \begin{tikzpicture}[xscale=.85, yscale=1]
      \tikzstyle{every node}=[font=\scriptsize]
      \tikzstyle{fade}=[opacity=0]

      \foreach[count=\i from 0] \x/\y/\stretch/\aleft/\aright/\bleft/\bright/\hd in {0/0/1/1/1/1/1/x, 6/1.5/.7/0/1/0/1/, 11/1.5/.7/0/1/1/0/, 6/-2.3/.7/1/0/0/1/, 11/-2.3/.7/1/0/1/0/} {
        \node[smallvertex] (rt\i) at (\x,\y) {};
        \nextnodelab{0\i}{rt\i}{-135:1.5*\stretch}{revarc}{.7}{\edgelab[-45]{2~~}{~}}
        \nextnodelab{00\i}{0\i}{-135:1.5*\stretch}{revarc}{.7}{\edgelab[-45]{5~~}{~}}
        \nextnodelab[leaf, label=below:$a$]{a\i}{00\i}{-90:2.121*\stretch}{revarc}{.5}{\edgelab{1~~}{~}}
        
        \nextnodelab[reti]{01\i}{0\i}{-45:1.5*\stretch}{revarc, opacity=\aleft}{.6}{\edgelab[45]{1}{\maybehide{\hd}{~0.3}}}
        
        \nextnodelab{001\i}{00\i}{-30:1.5*\stretch}{revarc}{.5}{\edgelab[45]{1~~}{~}}
        \nextnodelab[leaf, label=below:$b$]{b\i}{001\i}{-90:1.414*\stretch}{revarc}{.5}{\edgelab{2~~~}{}}
        
        \nextnodelab[reti]{0011\i}{001\i}{-30:1.5*\stretch}{revarc, opacity=\bleft}{.55}{\edgelab[60]{4}{\maybehide{\hd}{~0.4}}}
        \nextnodelab[leaf, label=below:$c$]{c\i}{0011\i}{-90:.707*\stretch}{revarc}{.7}{\edgelab{1~~}{~}}
        
        \nextnodelab{1\i}{rt\i}{-45:1.5*\stretch}{revarc}{.5}{\edgelab[45]{2~~}{~}}
        \nextnodelab{11\i}{1\i}{-45:2*\stretch}{revarc}{.5}{\edgelab[45]{1~~}{~}}
        \nextnodelab[leaf, label=below:$d$]{d\i}{11\i}{-115:1.5*\stretch}{revarc}{.7}{\edgelab[-45]{1~~}{}}
        \nextnodelab[leaf, label=below:$e$]{e\i}{11\i}{-65:1.5*\stretch}{revarc}{.5}{\edgelab[45]{}{~~2}}
        
        \draw[arc, opacity=\aright] (1\i) -- (01\i) node[pos=.3] {\edgelab[-45]{6}{\maybehide{\hd}{~0.7}}};
        \draw[arc, opacity=\bright] (01\i) -- (0011\i) node[pos=.25] {\edgelab[20]{~~8}{\maybehide{\hd}{~0.6}}};

        \node[opacity=\i] at ($(rt\i) + (195:2*\stretch)$) {$\sigma_\i$};
      }
      \foreach \u/\a/\la in {rt/0/\rho, 0/135/u, 01/180/r, 00/180/p, 001/-135/q, 0011/0/s, 1/45/v, 11/45/w}{
        \node at ($(\u0)+(\a:.3)$) {$\la$};
      }
    \end{tikzpicture}
    \caption{An example of a rooted phylogenetic network~$N$ (left) and its four switchings~$\sigma_1,\ldots ,\sigma_4$ (right). Edge lengths are indicated to the left and inheritance probabilities to the right of each edge. The probabilities of the switchings are $P(\sigma_1)=0.42$, $P(\sigma_2)=0.28$, $P(\sigma_3)=0.18$, $P(\sigma_4)=0.12$. The PD scores of $Z=\{a,c,e\}$ in the switchings are 
    $PD_{\sigma_1}=28$, $PD_{\sigma_2}=19$, $PD_{\sigma_3}=23$ and $PD_{\sigma_4}=19$. Hence, the APD score of~$Z$ is equal to $\APD(Z)=0.42\cdot 28+0.28\cdot 19+0.18\cdot 23+0.12\cdot 19=23.5$.
	}
    \label{fig:intro}
\end{figure}

\looseness=-1
%
Generalizing phylogenetic diversity from trees to networks is nontrivial,
and several network-based variants of PD have been proposed in recent years~\cite{bordewichNetworks,MAPPD23,van2024maximizing,van2025average,WickeFischer2018}.
%
In this paper, we focus on the \emph{average-tree phylogenetic diversity}~(APD)~\cite{van2025average,WickeFischer2018}, whose definition requires that
each edge~$uv$ incoming to a reticulation vertex~$v$ (a vertex with indegree greater than~$1$) is equipped with an ``inheritance probability''~$p_I(uv)$,
indicating the probability that any given gene of~$v$ has been inherited from~$u$.
Then, the APD is defined as the weighted average of phylogenetic diversities in all ``switchings''%
\footnote{\looseness=-1
A switching of $N$ is a tree resulting from the removal of all but one incoming edge for each reticulation.
For the purpose of defining the \APDlong, one can also take the weighted average of all embedded (``displayed'') phylogenetic trees,
but switchings are more convenient as they do not require vertex suppression.}
of the network,
see \cref{sec:prelims} for the definition and
see \cref{fig:intro} for an intuitive example.
We consider two main problems in this context:
\emph{computing} the APD of a given set of leaves and
\emph{deciding} whether there is a size-$k$ leaf set with at least a given APD.
Algorithms for the latter problem can be used to find a subset of the taxa with maximum APD.

%

\subparagraph*{Previous Work.}
Interestingly, it is already~\#P-hard to compute the APD of the set~$X_N$ of \emph{all taxa}~\cite{van2025average},
which is a special case of both problems, so it is unlikely that either can be solved in polynomial time.
On the positive side, there is an algorithm with running time~$2^{O(r)} \cdot |X_N|^{O(1)}$,
where 
$r$ is the number of reticulations~\cite{van2025average},
which can be used to solve both problems.


\subparagraph*{Our Results.}
\looseness=-1
Here, we consider parameters that are potentially much smaller than the number of reticulations and,
hence, may lead to much faster algorithms.
We generalize the tractability on trees by considering parameters measuring how ``close'' the input network is to being a tree.
The ``scanwidth'' is a fairly novel measure of tree-likeness introduced by
Berry, Scornavacca, and Weller~\cite{berry2020scanning},
which has received some attention recently~\cite{bruchhold2026exploiting,HoltgrefePANDA,holtgrefe2026tractable,MAPPD25,twvssw}.
This measure is of practical relevance, in particular when compared to the treewidth (of the underlying undirected graph).
First, it enables more intuitive dynamic programming, respecting the directionality of the given DAG and,
second, some problems that are {W[1]-hard} when parameterized by treewidth become fixed-parameter tractable when parameterized by scanwidth~\cite{twvssw}.
We start by showing in \cref{sec:sw-hardness} that it is \NP-hard to decide whether~$k$ leaves can attain a given diversity, even on networks with scanwidth~3.
In contrast, we show in \cref{sec:sw-DP} that the APD of the set~$X_N$ of all taxa can be computed in~$\Oh(2^{\sw} n)$~time,
and this result directly extends to arbitrary subsets of the taxa.

Orthogonally to considering networks with low scanwidth, we consider networks that are close to being ``reticulation-visible''.
Indeed, we show in \cref{sec:rv} that the APD of~$X_N$ can be computed in linear time on reticulation-visible networks,
implying that the APD of any set~$Z\subseteq X_N$ can be computed in linear time if the network ``induced\footnote{%
The subnetwork ``induced'' by a vertex-set~$Z$ is the result of removing all vertices that do not have a path to $Z$.}''
by~$Z$ is reticulation-visible.
Note that this generalizes the classic linear-time algorithm for trees by Faith~\cite{faith1992conservation}.
We generalize the result further, showing that
the APD of~$X_N$ can still be computed
in polynomial time on networks in which each biconnected component has only constantly many reticulations that are not visible.

Results marked with~$\star$ are deferred partly or completely to this paper's preprint and journal version.
	
	\section{Preliminaries}\label{sec:prelims}

  For a statement~$\psi$, we use Iverson brackets $\zeroone{\psi}$ to indicate the truth value of $\psi$,
  where $\zeroone{\psi}=1$ if $\psi$ is true and $0$~otherwise.
  We expect the reader not to confuse this notation with entries in our dynamic programming tables, which are denoted by $\ptable{.}{.}$ and $\dtable{.}{.}$.
  We use~$\uplus$ to denote a disjoint union.
  We generalize functions~$f: A \to \mathbb{Q}$ to subsets~$A' \subseteq A$ by~$f_\Sigma(A') := \sum_{a\in A'} f(a)$.

  \subparagraph*{Phylogenetic Networks.}\looseness=-1
  For any DAG~$N$, we let $V(N)$ and $E(N)$ denote the set of vertices and edges of $N$.
  If~$uv$ is an edge of $N$, then $u$ is a \emph{parent} of $v$
  and~$v$ is a \emph{child} of $u$.
  We call~$u$ and $v$ the \emph{tail} and \emph{head}, respectively, of $uv$.
  The number of parents and children of a vertex~$x$ in $N$ is its in-degree~$\indeg[N](x)$ and out-degree~$\outdeg[N](x)$, respectively.
  For a vertex or edge~$q$ of $N$, we let~$N-q$ denote the result of removing~$q$ from $N$
  (with all incident edges if $q$ is a vertex).
  %
  If a vertex~$u$ has a (possibly empty) path to a vertex~$v$ in $N$, that is, $N$ contains a $u$-$v$-path, then we write~$\haspath[N]{u}{v}$.
  We generalize this notion to the case where $u$ has a path to some vertex of a set of vertices~$Z$ and
  we further extend the notion to edges~$uv\in E(N)$, in which case $\haspath[N]{uv}{Z}$ means
  that $N$ contains a path~$p$ from $u$ to any vertex in $Z$ such that the first edge of $p$ is $uv$.

  In this work, a \emph{(phylogenetic) network}~$N$ is a rooted\footnote{There is exactly one source vertex called the \emph{root}.}
  directed acyclic graph (DAG)
  whose vertices all have
  in-degree~$\leq 1$ (\emph{tree-vertices}) or
  in-degree~$\geq 2$ and out-degree~$\leq 1$ (\emph{reticulations}) and
  whose \emph{leaves} (vertices with out-degree~0) are labeled bijectively by so-called taxa, which we refer to by~$X_N$.
  As they are in bijection, we use taxa and leaves interchangeably.
  If all vertices in $N$ have in-degree $\leq 1$, then~$N$~is a \emph{(phylogenetic) tree}.
  Note that we allow vertices with in-degree~1 and out-degree~1 since our algorithms can easily deal with such vertices.
  If~$v$ is a leaf and $\haspath{u}{v}$, then $v$ is an \emph{offspring} of $u$ and~$\off_{N}(u)$ denotes the set of offspring of~$u$.
  We generalize this notion to edges by~$\off_{N}(u_1u_2) := \off_{N}(u_2)$.
  If all paths from the root of $N$ to a leaf $x$ contain~$u$, then $u$ is \emph{visible from $x$}.
  Then, $N$ is \emph{reticulation-visible} if all reticulations of~$N$ are visible from some leaf.
  A \emph{tree-path} is a path whose inner vertices are tree-vertices.
  The network~$N$ is \emph{biconnected} if it does not contain cut-vertices
  (that is, $N-v$ is weakly connected for each vertex~$v$ of $N$),
  and a maximal biconnected subgraph of $N$ is called a \emph{biconnected component} of~$N$.
  We note here that all biconnected components of a rooted DAG are themselves rooted DAGs.
  The \emph{reticulation number} of a network is the total indegree of the reticulations minus the number of reticulations.
  The \emph{level} of a network is the maximum reticulation number of any biconnected component.
  %

  \subparagraph*{Scanwidth.}
  A \emph{tree-extension}~$\Gamma$ of a network~$N$ with leaf-set~$X_N$ is a tree whose vertices are $V(N)$ and whose leaves are $X_N$ and in which~$\haspath{u}{v}$ implies~$\haspath[\Gamma]{u}{v}$ for all~$u$ and~$v$.
  Intuitively, the ancestor-relation of $\Gamma$ extends the ancestor-relation of $N$.
  For any vertex~$v$ in $V(N)=V(\Gamma)$,
  let~$\GW(v) := \{uw \in E(N) \mid u\abov[\Gamma] v \aboveq[\Gamma] w\}$ denote the set of edges of~$N$ that ``pass'' $v$ in $\Gamma$.
  We omit the subscript if~$\Gamma$ is clear from the context.
  We abbreviate~$\GW(Z):=\bigcup_{z\in Z}\GW(z)$, for any set~$Z\subseteq V(N)$.
  To emphasize the similarity to treewidth, we sometimes call the sets~$\GW$ ``bags'' of the tree-extension.
  Then, the \emph{width} of~$\Gamma$ is $\sw(\Gamma):=\max_{v\in V(N)}|\GW(v)|$ and
  the \emph{scanwidth} of~$N$ is the minimum width over all tree-extensions of~$N$.
  See \cref{fig:scanwidth} for an example.

\begin{figure}[t]
    \centering
    \begin{tikzpicture}[xscale=-1, yscale=.85]
      \tikzstyle{fade}=[opacity=0]
      \colorlet{tegray}{lightgray!70!white}
      \tikzstyle{trext}=[line width=10pt, tegray]

      \node[smallvertex, label=below:$\rho$] (v0) at (0,0) {};
      \foreach[count=\i] \u/\st in {v/,u/,r/reti,p/,q/,s/reti} {
        \pgfmathtruncatemacro{\oldi}{\i - 1}
        \nextnode[vertex, \st, label=below:$\u$]{v\i}{v\oldi}{180:1}{fade}
      }
      \nextnode[leaf, label=below:$c$]{v7}{v6}{180:1}{fade}
      \nextnode[vertex, label=below:$w$]{v8}{v1}{135:1.5}{fade}
      \nextnode[leaf, label=0:$e$]{v9}{v8}{160:1}{fade}
      \nextnode[leaf, label=0:$d$]{v10}{v8}{-160:1}{fade}
      \nextnode[leaf, label=0:$a$]{v11}{v4}{135:1}{fade}
      \nextnode[leaf, label=0:$b$]{v12}{v5}{135:1}{fade}

      \foreach \u/\v/\b in {0/1/0, 0/2/-15, 1/3/15, 1/8/0, 2/3/0, 2/4/-15, 3/6/12, 4/5/0, 4/11/0, 5/6/0, 5/12/0, 6/7/0, 8/9/0, 8/10/0}{
        \draw[arc] (v\u) to[bend left=\b] (v\v);
      }

      \begin{pgfonlayer}{background}
        \draw[trext] (v0.west) -- (v7.east);
        \draw[trext] (v1.center) -- (v8.center) -- (v9.east);
        \draw[trext] (v1.center) -- (v8.center) -- (v10.east);
        \draw[trext] (v12.north east) -- (v5.center) -- (v0.center);
        \draw[trext] (v11.north east) -- (v4.center) -- (v0.center);
      \end{pgfonlayer}      
    \end{tikzpicture}
    \caption{A tree-extension~$\Gamma$ (in gray) of the phylogenetic network~$N$ in \cref{fig:intro}, which is drawn ``inside''~$\Gamma$, illustrating that the width of~$\Gamma$ is~$3$ (since $|\GW(r)|=|\{vr,ur,up\}|=3$) and $\sw\leq 3$.
	}
    \label{fig:scanwidth}
\end{figure}

\looseness=-1
  \subparagraph*{Switchings.}
  For a set~$R$ of vertices (usually reticulations) in a DAG~$N$,
  a \emph{partial switching}~$\sigma$ is a subgraph of $N$, resulting from removing all but one incoming edge of each~$r\in R$.
  The set of all switchings is denoted by~$\swis_R(N)$, for a set~$R$ of reticulations.
  We omit the prefix ``partial'' as well as the subscript~$R$, if $R$ contains all reticulations of~$N$.
  	
  \subparagraph*{Diversity.}
  In a network~$N$, reticulations represent evolutionary events when some species received genetic material from multiple parent species.
  That is, the gene pool of an (ancestral) hybrid vertex~$v$ consists of parts of the gene pools of the parents~$u_i$ of~$v$ in~$N$.
  We denote the probability that a uniformly sampled gene of $v$ was inherited from $u_i$ by $p_I(u_iv)$.
  Alternatively, one can interpret~$p_I(u_iv)$ as the probability that a gene inherited from~$u_i$ is picked when a gene of~$v$ is picked uniformly at random.
  Thus, we refer to $p_I$ as \emph{inheritance probability}.
  If~$\sum_{uv\in E(N)} p_I(uv)=1$ for all non-root vertices~$v$ of $N$, then $p_I$ is \emph{normal}.
  In particular, if $v$ is not the root nor a reticulation, then the unique incoming edge~$uv$ of $v$ has $p_I(uv)=1$.
  Given a network~$N$ with, root~$\rho$, edge-lengths~$\w:E(N)\to\mathbb{Q}_{\geq 0}$ and normal inheritance probabilities~$p_I$, 
  the \emph{probability} of a switching~$\sigma$
  is~$P(\sigma):=\prod_{v\in V(N)\setminus\rho} \sum_{uv\in E(\sigma)} p_I(uv)$.
  While the input network will have probability one, switchings of $N$ may have lower probabilities.
  For a set~$\swis'$ of switchings, 
  let~$P(\swis') := \sum_{\sigma\in\swis'} P(\sigma)$.
  The \emph{\APDlong} (APD) of a set~$Z$ of taxa is the weighted average of $\PD[\sigma](Z)$ for all switchings~$\sigma$ of~$N$, that is
  \begin{equation}\label{eq:APD}
    \APD(Z)
    := \;\smashsum{\sigma\in\swis(N)}\; P(\sigma) \cdot \PD[\sigma](Z)
    = \;\smashsum{\sigma\in\swis(N)}\; P(\sigma) \cdot \;\smashsum{e\in E(\sigma)}\; \w(e)\cdot \haspathzo[\sigma]{e}{Z}.
  \end{equation}
  Alternatively, this can be seen as the expected diversity~$\PD[\sigma](Z)$ of any switching~$\sigma$ of $N$
  drawn at random according to the probability~$P(\sigma)$.

  \looseness=-1
  By rearranging the sums in \Cref{eq:APD}, we can reformulate $\APD(Z)$ as the expected total weight of the edges with paths to leaves in~$Z$
  in a switching~$\sigma$ drawn at random from $\swis(N)$ according to the probability~$P(\sigma)$ (see \cite[Lemma~4.3]{van2025average}).
  To this end, we use the probability~$\edgeproba{uv}{Z}$
  that a randomly drawn switching of $N$ contains a path starting with $uv$ and ending in $Z$.
  \begin{equation}\label{eq:APD2025}
    \edgeproba{uv}{Z} = \;\smashsum{\sigma\in\swis(N)}\; P(\sigma)\cdot \haspathzo[\sigma]{uv}{Z}
    \qquad \text{so} \qquad
    \APD(Z) \stackrel{\eqref{eq:APD}}{=} \smashsum{uv \in E(N)} \w(uv)\cdot \edgeproba{uv}{Z}
  \end{equation}
  This work focuses on deciding whether any set of leaves of size at most~$k$ in $N$ has diversity at least~$D$.
  \problemdef{\MaxAPD}
  {A network $N$ with leaf-set~$X_N$, equipped with edge-weights $\w$ and normal inheritance probabilities~$p_I$, and integers~$k$ and $D$}
  {Is there a set~$Z \subseteq X_N$ of size at most~$k$ such that $\APD[N](Z)\geq D$?}

  \subparagraph*{Useful Observations.}

  Note that partial switchings can be ``completed'' by switching all remaining reticulations.
  We can observe that, for different switchings, the sets of such completions are disjoint and,
  thus, their contribution to the \APDlong is additive.

  \begin{observation}\label{obs:combine-swis}
    Let~$Z\subseteq X_N$ and 
    let~$\mathfrak{S}$ be a partition of $\swis(N)$.
    Then,
    \begin{equation*}
    	\APD[N](Z) = \sum_{\Sigma\in\mathfrak{S}}\sum_{\sigma\in\Sigma} P(\sigma)\cdot\PD[\sigma](Z).
    \end{equation*}
  \end{observation}

  \begin{observation}\label{obs:combine-reti-swis}
    Let~$Z$ be a set of taxa and
    let~$R$ be a set of reticulations of $N$.
    Then, $\{\swis(\sigma_R) \mid \sigma_R\in\swis_R(N)\}$ is a partition of $\swis(N)$ and
    $\APD[N](Z) = \sum_{\sigma_R\in\swis_R(N)} P(\sigma_R) \cdot \APD[\hat{\sigma}_R](Z)$, where $\hat{\sigma}_R$ is obtained from $\sigma_R$ by assigning inheritance probability~$1$ to edges entering vertices with indegree-$1$. 
  \end{observation}

	\section{Max-APD Parameterized by Scanwidth}\label{sec:sw-hardness}
  Many questions in algorithmic phylogenetics on networks become tractable when the input network is tree-like in the sense that it has a small scanwidth~\cite{Bruchhold24,MAPPD25} which is especially useful as scanwidth can be computed quite efficiently in practice~\cite{holtgrefe2026tractable,HvIJ24}.
  In this section, we show that this holds only partially for optimizing the \APDlong of an input network.
  In particular, it turns out to be \NP-hard to decide whether a given diversity can be reached by selecting $k$~leaves,
  even on binary networks with scanwidth~3,
  but it is still tractable on networks with scanwidth at most~2.
  On the other hand, it is possible to compute the \APDlong of the complete set~$X_N$ of all taxa in $\Oh(2^{s} n)$~time, given a width-$s$ tree-extension, even though this task is \#P-hard, in general~\cite{van2025average}.
	
  We first observe that \MaxAPD can be solved in polynomial time on networks with scanwidth at most~$2$.
	
  \begin{proposition}[\cite{holtgrefeThesis}]\label{prop:categorize-networks-sw=2}
		Let $N$ be a phylogenetic network.
    Then,
    \begin{enumerate}[(a)]
			\item the scanwidth of~$N$ is~$\leq 1$ if and only if~$N$ is a tree~\cite[Prop 3.30]{holtgrefeThesis} and
			\item the scanwidth of~$N$ is~$\leq 2$ if and only if~$N$ is a level-1-network~\cite[Prop 3.31]{holtgrefeThesis}.
		\end{enumerate}
	\end{proposition}
	
    We conjecture that \MaxAPD can be solved in polynomial time on networks with constant level since we believe that the algorithm showing fixed-parameter tractability for the number of reticulations~\cite{van2025average} can be 
    extended to the network's level by considering biconnected components independently.
    In combination with \cref{prop:categorize-networks-sw=2}, 
    this would imply that \MaxAPD is polynomial-time-solvable on networks with scanwidth at most~2.
    %
	%
	%
  Unfortunately, such a positive result can probably not be extended to include phylogenetic networks with scanwidth~3.
\newcommand{\propswNP}[1]{
	\begin{theorem}
		#1
		\MaxAPD remains \NP-hard on binary networks with scanwidth at most~3.
		\todo{JS: The result holds even on temporal networks. We should not mention that in a conference version, but incorporate it into the journal version.}
	\end{theorem}
}
\propswNP{\label{thm:const-sw}}
  \noindent
  To prove \cref{thm:const-sw}, we reduce the following variant of the Noah's ark problem to \MaxAPD.
    
  \problemdef{\ucNAP}
  {An edge-weighted binary phylogenetic tree~$T = (V,E)$ on leaf-set~$X_N$, survival-probabilities~$p_S: X_N\to(0,1]$, and integers~$k,D \in \NN$}
  {Is there a set~$S$ of~$k$ taxa such that the \emph{expected phylogenetic diversity}~$\EPD(S)$ is at least~$D$?}
  Herein, the expected phylogenetic diversity~$\EPD(S)$ is the sum of expected edge weights when selected taxa survive with their respective probability.
  More formally,
  \begin{equation}\label{eq:EPD}
    \EPD(S) := \sum_{e\in E} \w(e) \cdot \bigg( 1 - \;\smashoperator[lr]{\prod_{x\in S \, \cap \, \off_T(e)}}\; (1-p_S(x)) \bigg).
  \end{equation}
  We note that the survival probabilities in the formulation of \ucNAP capture the idea that,
  protecting a taxon does not necessarily guarantee its survival; the survival probabilities model this uncertainty.
  This leads to the notion of \emph{expected} surviving phylogenetic diversity of a set of leaves,
  which is not to be confused with the expected phylogenetic diversity score of a switching, used throughout the paper.
  
  Note that \ucNAP is
  \ifJournal
    \NP-hard~\cite{GNAP,van2024maximizing} even on trees of height 2 and ultrametric trees of height~3.
    It is easy to see that the hardness also holds for binary trees.
  \else
    \NP-hard~\cite{GNAP,van2024maximizing} and
    it is easy to see that the hardness also holds for binary trees.
  \fi
  We use this in the following reduction.
  \begin{figure}[t]
    \centering
    \begin{tikzpicture}
      \node[smallvertex, label=10:{$\rho=\rho'$}] (rt) at (0,0) {};

      \foreach[count=\i] \otherchild/\weight/\pri in {\beta/2M/', \alpha/1/} {
      \node at ($(rt)+(.5,-.5)+(2.5*\i-3.8,0)$) {$T\pri$};
        \coordinate (LL\i) at ($(rt) + (-150:1.5) + (-30:\i/1.5)$);
        \draw[fill=gray, opacity=0.6] (rt) -- (LL\i) -- ($(LL\i) + (2+\i/2,0)$) -- (rt);
        \node[smallvertex] (x\i) at ($(LL\i)+(.3,0)$) {};
        \node at ($(x\i) + (-300 + 140*\i:.4)$) {$x\pri$};
        \node[smallvertex, label=180*\i:{$\otherchild_x$}] (b\i) at ($(x\i) + (180+60*\i:1)$) {} edge[revarc] node[xshift=-21 + 14*\i] {{\scriptsize $\weight$}} (x\i);
      }

      \node[smallreti, label=right:$\gamma_x$] (gamma) at ($(x2)+(-120:.75)$) {} edge[revarc] (x1) edge[revarc] (x2);
      \node[smallleaf, label=left:$\delta_x$] (delta) at ($(gamma)+(-120:.75)$) {} edge[revarc] node[xshift=12] {{\scriptsize $M-1$}} (gamma);
    \end{tikzpicture}
    \caption{Illustration of \cref{cons:const-sw}. Semitransparent gray trees are copies of the input tree~$T$ of \ucNAP.
    Edge-weights are written next to the edges and
    the two reticulation edges have inheritance probability~$p_S(x)$ and $1-p_S(x)$, respectively.}
    \label{fig:const-sw}
  \end{figure}

  \begin{construction}\label{cons:const-sw}
    Given an instance~$\Instance = (T=(V,E),p_S,k,D)$ of \ucNAP with
    survival probabilities~$p_S$,
    leaf-set~$X_N$,
    root~$\rho$ and
    $D>1$,
    we construct the following instance~$\Instance':=(N, k', D')$ of \MaxAPD.
    First, we construct a network~$N$ with inheritance probability~$p_I$.
    To this end,
    let~$\rho':=\rho$,
    let~$M := \w_\Sigma(E)$, and
    duplicate $T$ in the following sense:
    \begin{enumerate}
      \item
      Add a vertex~$v'$ for each vertex~$v \in V\setminus\{\rho\}$, let~$\rho'=\rho$,
        and add an edge~$u'w'$ for each edge~$uw \in E$, setting~$\w(u'w'):=\w(uw)$.
        Let $E'$ denote the resulting edge-set.
      \item
      For each taxon~$x$,
      add a gadget of four vertices~$Y_x:=\{\alpha_x, \beta_x, \gamma_x, \delta_x\}$ and
        edges~$x\gamma_x$ and $x'\gamma_x$ with weight~1 and inheritance probabilities~$p_I(x\gamma_x):=p_S(x)$ and $p_I(x'\gamma_x):=1-p_S(x)$,
        and edges~$x\alpha_x$ with weight~1, $\gamma_x\delta_x$ with weight~$M-1$, and~$x'\beta_x$ with weight~$2M$.
        Note that $\gamma_x$ is a reticulation with parents~$x$ and $x'$.
    \end{enumerate}
    Finally, set~$k' := k+|X_N|$ and~$D' := D + M\cdot (k + 2\cdot|X_N| + 1)$.
  \end{construction}

\newcommand{\proofPropswNP}[2]{
  \begin{proof}#1
    \looseness=-1
    Let~$\Instance':=(N, k', D')$ be an instance of \MaxAPD resulting
    from applying \cref{cons:const-sw} to an instance~$\Instance = (T=(V,E),k,D)$ of \ucNAP
    where~$T$ is a binary tree with leaf-set~$X_N$.
    Let $E'$ be defined as in \cref{cons:const-sw}.
    Then, it is easy to see that $N$ is binary after refining its root.
    To show that $N$ has scanwidth three, we construct a tree-extension $\tau$ for~$N$ by
    first starting with~$T$ itself,
    then subdividing each edge~$uv$ with a new vertex~$v'$ and,
    finally, for each leaf~$x$ adding child~$\beta_x$ to~$x'$, children~$\alpha_x$ and~$\gamma_x$ to~$x$, and child~$\delta_x$ to~$\gamma_x$.
    Since $T$ is a tree, for each vertex~$v\in V(T)$ with parent~$u$ and children~$w_1$ and~$w_2$,
    the bag~$\GW(v')$ contains two edges, $uv$ and~$u'v'$,
    and~$\GW(v)$ contains three edges, $uv$, $v'w_1'$, and~$v'w_2'$.
    For each leaf $x$ with parent~$v$, the bag~$\GW(x')$ contains edges~$v'x'$ and~$vx$, $\GW(x)$ contains~$vx$ and~$x'\gamma_x$, $\GW(\gamma_x)$ contains~$x'\gamma_x$ and~$x\gamma_x$, the other three bags contain only their incoming edge.
    Consequently, $\tau$ has scanwidth~3.
	
	#2
  \end{proof}
}
\newcommand{\proofPropswNPfull}{
    \medskip\noindent
    It remains to show that $\Instance \in \ucNAP$ if and only if~$\Instance' \in \MaxAPD$.

    ``$\Rightarrow$'':
		Let~$S := \{x_1,\dots,x_k\}$ be a solution for \ucNAP on~$\Instance$ and
    let~$S' := \{ \delta_x \mid x\in S\} \cup \{ \beta_x \mid x\in X_N \}$,
    noting that $|S'|=k+|X_N| = k'$.
		We will consider edges individually.
    First, the edges incoming to~$S'$ contribute~$k\cdot (M-1) + |X_N|\cdot 2M$ to the overall diversity.
    Further, each edge in~$E'\setminus E$ has a path without any reticulation to a leaf in~$\{ \beta_x \mid x\in X_N \} \subseteq S'$,
    totalling a score of~$\sum_{e\in E}\w(e) = M$.
    Further, for each $x\in S$, the two edges incoming to~$\gamma_x$ contribute~$p_S(x) + (1-p_S(x)) = 1$, totalling a score of $k$.
    Thus, the total contribution so far is~$k\cdot (M-1) + |X_N|\cdot 2M + M + k = M\cdot (k + 2|X_N| + 1)$.
		
    The remaining edges are the edges in $E$.
    For these edges~$e$, we have~$\nhaspath[\sigma]{e}{Z}$ if $\sigma$ does not contain~$x\gamma_x$ for any~$x\in S\cap\off_T(e)$,
    the probability of which is~$\prod_{x\in S \,\cap\, \off_T(e)} (1-p_S(x))$.
		We thus conclude
		\begin{align*}
			\APD(S')
      \stackrel{\eqref{eq:APD2025}}{=}&~~ \smashsum{e \in E(N)} \w(e)\cdot \smashsum{\text{switching }\sigma}\; P(\sigma) \cdot \haspathzo[\sigma]{e}{S'}\\
      =&~~ M\cdot (k + 2\cdot|X_N| + 1) + \sum_{e \in E} \w(e)\cdot \smashsum{\text{switching }\sigma}\; P(\sigma) \cdot \haspathzo[\sigma]{e}{S'}\\
      \stackrel{\sum P(\sigma)=1}{=} &~~ M\cdot (k + 2\cdot|X_N| + 1) + \sum_{e \in E} \w(e)\cdot \left(1- \smashsum{\text{switching }\sigma}\; P(\sigma) \cdot \nhaspathzo[\sigma]{e}{S'}\right)\\
      =&~~ M\cdot (k + 2\cdot|X_N| + 1) + \sum_{e \in E} \w(e)\cdot \left(1 - \smashoperator[lr]{\prod_{x \in S \,\cap\, \off(e)}}\; (1 - p_S(x))\right)\\
      \stackrel{\eqref{eq:EPD}}{=}&~~ M\cdot (k + 2\cdot|X_N| + 1) + \EPD(S)\\
      \ge&~~ M\cdot (k + 2\cdot|X_N| + 1) + D = D'.
		\end{align*}
		
    \medskip
    ``$\Leftarrow$'':
		In the converse direction,
    let~$S$ be a solution of \MaxAPD on~$\Instance'$,
    which contains
    \begin{itemize}
      \item $a_S$ taxa of~$A_N := \{ \alpha_x \mid x\in X_N \}$,
      \item $b_S$ taxa of~$B_N := \{ \beta_x \mid x\in X_N \}$, and
      \item $d_S$ taxa of~$D_N := \{ \delta_x \mid x\in X_N \}$.
    \end{itemize}
    Then, $D+M\cdot (k+2\cdot |X_N|+1) = D' \le \APD(S) \le 2M\cdot b_S + M\cdot d_S + 2\w_\Sigma(E) + a_S = a_S + M \cdot (2b_S + d_S + 2)$.
    Now, if $b_S+d_S < |X_N| + k$, then the right-hand-side is at most $1 + M\cdot (2|X_N| + k + 1) < D'$ since $D > 1$.
    But then, since~$a_S + b_S + d_S = |S| \le k' = k+|X_N|$ and $\max\{a_S,b_S,d_S\}\leq |X_N|$, we conclude that $b_S=|X_N|$ and $d_S=k$ and $a_S=0$.
		We claim that~$S' := \{ x \mid \delta_x \in S \}$ is a solution of~\ucNAP on~\Instance.
		Clearly, $|S'| = |S\cap D_N| = d_S = k$ and
    $\APD(S)$ can be partitioned into
    \begin{itemize}
      \item $b_S\cdot 2M$, scored on the edges~$x'\beta_x$,
      \item $d_S\cdot (M-1)$, scored on the edges~$\gamma_x\delta_x$,
      \item $d_S$, scored on the edges~$x'\gamma_x$ and $x\gamma_x$, exactly one of which is in every switching,
      \item $\w'_\Sigma(E')$, scored on the edges of $T'$, since $\beta_x\in S$ for all~$x\in X_N$, and
      \item $\EPD[T](S')$, scored on the edges of $T$ since $\alpha_x\notin S$ for all $x\in X_N$.
    \end{itemize}
    Thus, $\EPD[T](S') = \APD(S) - b_S\cdot 2M - d_S \cdot M - \w'_\Sigma(E') \geq D' - M\cdot (2|X_N| + k + 1) = D$.
}
\proofPropswNP{[Proof of~\Cref{thm:const-sw}]}{\proofPropswNPfull}

	\section{Computing \APDX Parameterized by Scanwidth}\label{sec:sw-DP}
  While \cref{thm:const-sw} shows that optimizing the \APDlong is probably intractable with respect to the scanwidth of the input,
  we show that computing the \APDlong of any fixed set~$Z \subseteq X_N$ is FPT in the with of a given tree-extension.
  To this end, fix a set of taxa~$Z$ and an edge~$uv$ and
  observe that~$\nhaspath[\sigma]{uv}{Z}$ for every switching~$\sigma$, if~$\off(v) \cap Z = \emptyset$.
  Consequently, we can delete all vertices~$v$ with~$\off(v) \cap Z = \emptyset$ from~$N$.
  This lets us assume that $N$ contains no such vertex and computing $\APD[N](X_N)$ is, thus, sufficient.
  
\newcommand\thmswfpt[1]{
	\begin{theorem}[$\star$]
		#1
		Let~$N$ be a network on~$X_N$ with $n$~vertices.
		Given a tree-extension~$\Gamma$ for~$N$,
		\APDX can be computed in~$\Oh(2^{w_\Gamma} n)$~time,
		where~$w_\Gamma$ is the width of $\Gamma$.
	\end{theorem}
}
\thmswfpt{\label{thm:FPT-sw}}

\noindent
We prove \cref{thm:FPT-sw} using a bottom-up dynamic programming over the given tree-extension~$\Gamma$.

\subparagraph*{Table Semantics.}
For each vertex~$v$, we construct two main (and four auxiliary) tables indexed by subsets~$Y\subseteq\GW(v)$.
To define them, we use the following notation.

\begin{definition}[compatible switching]\label{def:compatible}
  Let~$v\in V(N)$,
  let~$A\subseteq V(N)$ and
  let~$Y \subseteq \GW(v) \cap \GW(A)$.
  A switching $\sigma\in\swis(N)$ is \emph{$(v,A,Y)$-compatible}
  if~$\haspath[\sigma]{uw}{X_N} \iff uw\in Y$ holds for each $uw\in\GW(v)\cap\GW(A)$.
  In the special case that $A=\{v\}$, we say $\sigma$ is \emph{$(v,Y)$-compatible}.
\end{definition}

Intuitively, a switching~$\sigma$ is $(v,Y)$-compatible if exactly the ``selected'' edges~$Y$ of $\GW(v)$ can reach leaves in~$\sigma$.
In the running example of \cref{fig:intro,fig:scanwidth},
$\sigma_1$ is $(r,\{up,vr\})$-compatible,
$\sigma_3$ is $(r,\{up,ur\})$-compatible, and
$\sigma_2$ and $\sigma_4$ are $(r,\{up\})$-compatible.

Now, the first table~$p$ with entries~$\ptable{v}{Y}\in \mathbb{Q} \cap [0,1]$ stores the probability of drawing a $(v,Y)$-compatible switching,
when each switching~$\sigma$ has probability~$P(\sigma)$.
The second table~$d$ with entries~$\dtable{v}{Y} \in \mathbb{Q}_{\ge 0}$ stores the overall \APDlong of these switchings,
restricted to edges below~$\GW(v)$ (that is, assuming that all edges that are not on a path from an edge of $\GW(v)$ to a leaf have weight zero).
Then, the \APDlong of $N$ can be read from entry $\dtable{\rho}{\emptyset}$, where~$\rho$ is the root of~$N$ and $\Gamma$.

\subparagraph*{Table Computation.}
In the following, we show how to compute the table entries.
For leaves~$x\in X_N$, we have $\GW(x)=\{ux\}$ for some~$u$, and we compute
\begin{equation}\label{eq:sw-base}
  \ptable{x}{Y} :=
  \begin{cases}
    1 & \text{if $Y=\{ux\}$}\\
    0 & \text{if $Y=\emptyset$}
  \end{cases}
  \qquad
  \text{and}
  \qquad
  \dtable{x}{Y} :=
  \begin{cases}
    \w(ux) & \text{if $Y=\{ux\}$}\\
    0 & \text{if $Y=\emptyset$}
  \end{cases}
\end{equation}
This is correct since all switchings contain each leaf's unique incoming edge.

Now, let~$v$ be an internal vertex of~$N$, either a tree-vertex or a reticulation,
and let~$w_1,\dots,w_\Delta$ be the children of~$v$ in the tree-extension~$\Gamma$.%
\footnote{Note that~$v$ can have more than~$\Delta$ children in~$N$,
but we can assume that the subtree of~$\Gamma$ below each~$w_i$ contains at least one child of $v$ in $N$ (see~\cite{berry2020scanning}).}
For a more concise presentation, let us agree on the following abbreviations:
\begin{enumerate}[(a)]
  \item
    The sets of incoming and outgoing edges of~$v$ in~$N$ are~$\inedges$ and~$\outedges$, respectively.
  \item
    For each~$h\leq\Delta$ and each~$Z\subseteq E(N)$, including~$Y$ and~$\outedges$, we write~$Z_h:=Z\cap\GW(w_h)$
    and $Z_{\leq h}:=\bigcup_{i=1}^h Z_i$.
    In particular, $\outedges_h := \outedges \cap \GW(w_h)$.
  \item
    For any~$h\leq\Delta$ and any set~$Z\subseteq \outedges$,
    the probability of drawing a switching~$\sigma$ such that
    each $\sigma$ is $(w_i,Y_i\uplus Z_i)$-compatible for each~$i\le h$,
    is~$\childprob[h]{Z}:=\prod_{i=1}^h \ptable{w_i}{Y_i\uplus Z_i}$.
\end{enumerate}
Intuitively, $\childprob[h]{Z}$ is the probability of drawing a switching that is compatible with~$Y$ and~$Z$ for all of the first~$h$ children of~$v$ in~$\Gamma$.
To compute $\ptable{v}{Y}$ and $\dtable{v}{Y}$,
we distinguish cases according to how many of $v$'s incoming edges are in $Y$.
%
If $|Y\cap \inedges|\geq 2$, then no switching can be $(v,Y)$-compatible as $v$ only has one incoming edge in any switching,
so we store~$\dtable{v}{Y}=\ptable{v}{Y}=0$.

If~$Y\cap \inedges=\emptyset$, then~$v$ cannot reach any leaf in any $(v,Y)$-compatible switching.
In this case, we just combine the probabilities and expected diversity scores of the children~$w_i$ in $\Gamma$ as follows, assuming $\ptable{w_i}{Y_i}\neq 0$:%
\footnote{If $\ptable{w_i}{Y_i}=0$, also $\ptable{v}{Y}=\childprob{\emptyset}=0$,
and $Y_i$ is an illegal choice for~$w_i$. This lets us safely ignore $\dtable{v}{Y}$ in this case.}
\vspace{-2ex}
\begin{equation}
  \ptable{v}{Y}
  := \childprob{\emptyset}
  \qquad \text{and} \qquad
  \dtable{v}{Y}
  := \childprob{\emptyset} \cdot \sum_{i=1}^\Delta \frac{\dtable{w_i}{Y_i}}{\ptable{w_i}{Y_i}}     \label{eq:sw-internal-pd-empty}
\end{equation}

\looseness=-1
Finally, the most interesting case is~$Y\cap \inedges = \{uv\}$, for some~$u$.
Then, at least one outgoing edge of~$v$ reaches a leaf in all $(v,Y)$-compatible switchings.
There are different switchings for each subset of surviving outgoing edges of~$v$ in~$N$,
all of which influence the table entries $\ptable{v}{Y}$ and $\dtable{v}{Y}$.
To compute these entries, we thus have to consider all sets~$Z$ of outgoing edges of~$v$ in~$N$.
Then, we define the table entries
\begin{align}
  \ptable{v}{Y}
  & := p_I(uv) \cdot \;\smashsum{Z\subseteq \outedges;~Z \neq \emptyset} \; \childprob{Z}    \label{eq:sw-internal-p}\\
  \dtable{v}{Y},
  & := \w(uv) \cdot \ptable{v}{Y} +
      p_I(uv) \cdot 
      \;\smashsum{Z\subseteq \outedges;~Z \neq \emptyset}\;
      \childprob{Z} \cdot
      \sum_{i=1}^\Delta
      \frac{\dtable{w_i}{Y_i\uplus Z_i}}{\ptable{w_i}{Y_i\uplus Z_i}},   \label{eq:sw-internal-d}
\end{align}
again ignoring table entries where the denominator is~0.
%
Unfortunately, this formulation requires iterating over all subsets~$Z$ of out-edges of~$v$ in~$N$,
which might exceed the target running time.
The standard technique in this case is a dynamic program over the out-edges~$vw_i$ of~$v$ in~$\Gamma$,
which is not necessary for the correctness of the algorithm, but to achieve the desired running time.
%
%
To this end, we partition the set of switchings again,
into the switchings~$\sigma$ that contain a path from $v$ to a leaf below any~$w_i$ with $i\leq h$ (tables~$C\mathrm{p}$ and $C\mathrm{d}$),
and those that do not (tables~$Q\mathrm{p}$ and $Q\mathrm{d}$).
Formally, the table entries for $h\geq 1$ are
\begin{align}
  Q\ptable{v}{Y,h} := \childprob[h]{\emptyset}
  & \text{and} \qquad
  Q\dtable{v}{Y,h}
   := \childprob[h]{\emptyset} \cdot \sum_{i=1}^h \frac{\dtable{w_i}{Y_i}}{\ptable{w_i}{Y_i}}   \label{eq:sw-auxQ-pd}\\
  C\ptable{v}{Y,h}
   := \;\;\;\smashsum{Z\subseteq E_{v,\leq h}^+;~Z\neq\emptyset}\;\;  \childprob[h]{Z}
  & \text{and} \qquad
  C\dtable{v}{Y,h}
   := \;\;\;\smashsum{Z\subseteq E_{v,\leq h}^+;~Z\neq\emptyset} \;\;
    \childprob[h]{Z} \cdot~~
    \sum_{i=1}^h
      \frac{\dtable{w_i}{Y_i\uplus Z_i}}{\ptable{w_i}{Y_i\uplus Z_i}}    \label{eq:sw-auxC-pd}
\end{align}
Then, since $\outedges\subseteq\biguplus_{i=1}^\Delta\GW(w_i)$, we can read what we need for~\eqref{eq:sw-internal-p} and~\eqref{eq:sw-internal-d}
from $C\ptable{v}{Y,\Delta}$ and $C\dtable{v}{Y,\Delta}$.
In order to formulate a recurrence, we define tables~$C$ and~$Q$ starting with $h=0$:
\begin{equation}
  Q\ptable{v}{Y,0} := 1
  \qquad \text{and} \qquad
  C\ptable{v}{Y,0} := C\dtable{v}{Y,0} := Q\dtable{v}{Y,0} := 0
  \label{eq:sw-aux-pd0}
\end{equation}
Intuitively, $Q\ptable{v}{Y,h}$ and $Q\dtable{v}{Y,h}$ represent switchings~$\sigma$
in which $v$ does not have a path to a leaf below~$w_i$ for any $i\leq h$.
Thus, the probabilities of the sets of switchings simply multiply for each~$i$.
Further, the expected diversity of such a switching can be computed by combining
the expected diversity of a combined switching ``below $w_1,\ldots,w_{h-1}$'' in which $v$ does not reach a leaf (represented by $Q\dtable{v}{Y,h-1}$) and
a switching ``below $w_h$'' in which $v$ does not reach a leaf (represented by $\dtable{w_h}{Y_h}$).
We can now show how to compute the tables for $0<h\leq\Delta$ recursively.

\begin{lemma}\label{lem:sw-auxQ-correct}
  Let $0<h\leq\Delta$.
  Then,
  \begin{align*}
    Q\ptable{v}{Y, h} & = Q\ptable{v}{Y, h-1} \cdot \ptable{w_h}{Y_h}\\
    Q\dtable{v}{Y, h} & = Q\dtable{v}{Y, h-1} \cdot \ptable{w_h}{Y_h} + \dtable{w_h}{Y_h}\cdot Q\ptable{v}{Y,h-1}.
  \end{align*}
\end{lemma}
\begin{proof}
  For~$h=1$, observe that
  \begin{align*}
    Q\ptable{v}{Y,1}
    & \stackrel{\eqref{eq:sw-auxQ-pd}}{=}
      \childprob[1]{\emptyset}
    = \ptable{w_1}{Y_1}
    \stackrel{\eqref{eq:sw-aux-pd0}}{=}
      \underbrace{Q\ptable{v}{Y,0}}_{=1} \cdot \ptable{w_1}{Y_1}\\
    Q\dtable{v}{Y,1}
    & \stackrel{\eqref{eq:sw-auxQ-pd}}{=} 
        \childprob[1]{\emptyset} \cdot \frac{\dtable{w_1}{Y_1}}{\ptable{w_1}{Y_1}}
    = \dtable{w_1}{Y_1}
    \stackrel{\eqref{eq:sw-aux-pd0}}{=}
      \underbrace{Q\dtable{v}{Y, 0}}_{=0} \cdot \ptable{w_1}{Y_1} + \dtable{w_1}{Y_1}\cdot \underbrace{Q\ptable{v}{Y,0}}_{=1}.\\
  \intertext{For $h>1$, we have}
    Q\ptable{v}{Y,h}
    & \stackrel{\eqref{eq:sw-auxQ-pd}}{=}
      \childprob[h]{\emptyset}
    = \prod_{i=1}^h\ptable{w_i}{Y_i}
    = \childprob[h-1]{\emptyset} \cdot \ptable{w_h}{Y_h}
    \stackrel{\eqref{eq:sw-auxQ-pd}}{=}
      Q\ptable{v}{Y, h-1} \cdot \ptable{w_h}{Y_h}\\
    Q\dtable{v}{Y,h}
    & \stackrel{\eqref{eq:sw-auxQ-pd}}{=}
      \childprob[h]{\emptyset} \cdot \sum_{i=1}^h \frac{\dtable{w_i}{Y_i}}{\ptable{w_i}{Y_i}}
      = \childprob[h-1]{\emptyset} \cdot \ptable{w_h}{Y_h} \cdot
        \left(
          \frac{\dtable{w_h}{Y_h}}{\ptable{w_h}{Y_h}}
          + \sum_{i=1}^{h-1} \frac{\dtable{w_i}{Y_i}}{\ptable{w_i}{Y_i}}
        \right)\\
    & \stackrel{\eqref{eq:sw-auxQ-pd}}{=}
      \childprob[h-1]{\emptyset} \cdot \dtable{w_h}{Y_h}
      +
      \ptable{w_h}{Y_h} \cdot Q\dtable{v}{Y,h-1}\\
    & \stackrel{\eqref{eq:sw-auxQ-pd}}{=}
      Q\ptable{v}{Y,h-1} \cdot \dtable{w_h}{Y_h}
      +
      \ptable{w_h}{Y_h} \cdot Q\dtable{v}{Y,h-1}\qedhere
  \end{align*}
\end{proof}

\looseness=-1
Intuitively, to compute the $C$-tables for $h$,
we consider switchings~$\sigma$ for which~$v$ reaches (in $\sigma$) some leaf ``below'' one of the first $h$~children~$w_1,\ldots,w_h$ of~$v$ (in $\Gamma$).
We partition the set of these switchings according to whether or not~$v$~reaches a leaf below~$w_h$.
Then, we express their combined probability (and diversity) in terms of the probabilities of
drawing a switching~$\sigma_h$ below~$w_h$ (represented by~$\ptable{w_h}{Y_h\uplus Z}$ for some~$Z\subseteq\outedges_h$)
and a switching~$\sigma_{\leq h-1}$ below~$w_1,\ldots,w_{h-1}$ (represented by~$C\ptable{v}{Y,h-1}$ if $v$~reaches a leaf below~$w_1,\ldots,w_{h-1}$,
and by~$Q\ptable{v}{Y,h-1}$, otherwise).
Herein, we have to ensure that $v$ reaches a leaf in~$\sigma_h$ or $\sigma_{\leq h-1}$ (since we are computing a $C$-table).

\newcommand\lemswaux[1]{
  \begin{lemma}
  	#1
    Let $0<h\leq\Delta$ and
    abbreviate
    $p^{C+D}_{h-1}:=C\ptable{v}{Y,h-1} + Q\ptable{v}{Y, h-1}$ and
    $d^{C+D}_{h-1}:=C\dtable{v}{Y,h-1} + Q\dtable{v}{Y, h-1}$.
    Then,
    \begin{align*}
      C\ptable{v}{Y, h}
      = & C\ptable{v}{Y, h-1}\cdot\ptable{w_h}{Y_h} 
        + \;\;\;\smashsum{Z\subseteq \outedges_h;~Z\neq\emptyset}\;\; p^{C+Q}_{h-1} \cdot \ptable{w_h}{Y_h\uplus Z}\\
      C\dtable{v}{Y, h}
      = & C\dtable{v}{Y, h-1} \cdot\ptable{w_h}{Y_h}
        + \dtable{w_h}{Y_h} \cdot C\ptable{v}{Y,h-1}\\
      & + \;\;\;\smashsum{Z\subseteq \outedges_h;~Z\neq\emptyset}\;\;
          \dtable{v}{Y_h\uplus Z}\cdot p^{C+Q}_{h-1}
        + d^{C+Q}_{h-1} \cdot \ptable{w_h}{Y_h\uplus Z}
    \end{align*}
  \end{lemma}
}
\lemswaux{\label{lem:sw-auxC-correct}}
  \begin{proof}
    First, for $h=1$, observe that
    \begin{align*}
      C\ptable{v}{Y, 1}
      \stackrel{\eqref{eq:sw-auxC-pd}}{=} &
      \;\;\;\smashsum{Z\subseteq E_{v,\leq 1}^+;~Z\neq\emptyset}\;\; \overbrace{\ptable{w_1}{Y_1\cup Z_1}}^{=\childprob[1]{Z}}
      \stackrel{\eqref{eq:sw-aux-pd0}}{=} 
        \overbrace{C\ptable{v}{Y, 0}}^{=0}\cdot\ptable{w_1}{Y_1} \\
      & + \;\;\;\smashsum{Z\subseteq \outedges_1;~Z\neq\emptyset}\;\; \overbrace{(C\ptable{v}{Y,0} + Q\ptable{v}{Y,0})}^{=p^{C+Q}_0 = 1} \cdot \ptable{w_1}{Y_1\cup Z_1}\\
      C\dtable{v}{Y, 1}
      \stackrel{\eqref{eq:sw-auxC-pd}}{=} &
        \;\;\;\smashsum{Z\subseteq E_{v,\leq 1}^+;~Z\neq\emptyset} \;\;
        \overbrace{\ptable{w_1}{Y_1\cup Z_1}}^{=\childprob[1]{Z}} \cdot
        \frac{\dtable{w_1}{Y_1\cup Z_1}}{\ptable{w_1}{Y_1\cup Z_1}}
      \;\; = \;\; \;\;\;\smashsum{Z\subseteq \outedges_1;~Z\neq\emptyset} \;\;
        \dtable{w_1}{Y_1\cup Z_1}\\
      = & \overbrace{C\dtable{v}{Y, 0}}^{=0} \cdot\ptable{w_1}{Y_1}
        + \dtable{w_1}{Y_1} \cdot \overbrace{C\ptable{v}{Y,0}}^{=0}\\
      & + \;\;\;\smashsum{Z\subseteq \outedges_1;~Z\neq\emptyset}\;\;
          \dtable{w_1}{Y_1\cup Z_1}\cdot \overbrace{p^{C+Q}_0}^{=1}
          + \overbrace{d^{C+Q}_0}^{=0} \cdot \ptable{w_1}{Y_1\cup Z_1}.\\
    \intertext{For $h>1$, we have}
      C\ptable{v}{Y, h}
      \stackrel{\eqref{eq:sw-auxC-pd}}{=} &
        \;\;\;\smashsum{Z\subseteq E_{v,\leq h}^+;~Z\neq\emptyset} \childprob[h]{Z}
        =
        \;\;\;\smashsum[l]{Z\subseteq E_{v,\leq h}^+;~Z\neq\emptyset}\;  \prod_{i=1}^{h}\ptable{w_i}{Y_i\uplus Z_i}
        =
        \;\;\;\smashsum{Z\subseteq E_{v,\leq h}^+;~Z\neq\emptyset}\;  \ptable{w_h}{Y_h\uplus Z_h} \cdot \prod_{i=1}^{h-1}\ptable{w_i}{Y_i\uplus Z_i}\\
      = &
        \;\;\;\smashsum{A\subseteq E_{v,\leq h-1}^+;~B\subseteq \outedges_h;~A\cup B\neq\emptyset}\;  \ptable{w_h}{Y_h\cup B} \cdot \prod_{i=1}^{h-1} \ptable{w_i}{Y_i\cup A_i}
        =
        \;\;\;\smashsum{A\subseteq E_{v,\leq h-1}^+;~B\subseteq \outedges_h;~A\cup B\neq\emptyset}\;  \ptable{w_h}{Y_h\cup B} \cdot \childprob[h-1]{A}\\
      = & 
        \ptable{w_h}{Y_h \cup \emptyset} \cdot \smashsum[l]{A\subseteq E_{v,\leq h-1}^+;~A\neq\emptyset} \childprob[h-1]{A} +
        \;\;\;\smashsum{B\subseteq \outedges_h;~B\neq\emptyset}\;  \ptable{w_h}{Y_h\cup B} \cdot \smashsum{A\subseteq E_{v,\leq h-1}^+} \childprob[h-1]{A}\\
      = & 
        \ptable{w_h}{Y_h} \cdot C\ptable{w_i}{Y_i,h-1} +
        \;\;\;\smashsum{B\subseteq \outedges_h;~B\neq\emptyset}\;  \ptable{w_h}{Y_h\cup B} \cdot \left( \prod_{i=1}^{h-1} \childprob[h-1]{\emptyset} + \smashsum[l]{A\subseteq E_{v,\leq h-1}^+;~A\neq\emptyset} \prod_{i=1}^{h-1} \childprob[h-1]{A} \right)\\
      = & 
        \ptable{w_h}{Y_h} \cdot C\ptable{w_i}{Y_i,h-1} +
        \;\;\;\smashsum{B\subseteq \outedges_h;~B\neq\emptyset}\;  \ptable{w_h}{Y_h\cup B} \cdot P^{C+Q}_{h-1}\\
      \intertext{And for the second table of~$C$ consider}
      C\dtable{v}{Y,h}
      \stackrel{\eqref{eq:sw-auxC-pd}}{=} &
        \smashsum{Z\subseteq E_{v,\leq h}^+;~Z\neq\emptyset} \;\;
        \childprob[h]{Z} \cdot \sum_{i=1}^h
        \frac{\dtable{w_i}{Y_i\uplus Z_i}}{\ptable{w_i}{Y_i\uplus Z_i}}
      = 
        \smashsum{Z\subseteq E_{v,\leq h}^+;~Z\neq\emptyset} \;\;
        \childprob[h-1]{Z} \cdot \ptable{w_h}{Y_h\uplus Z_h} \cdot \sum_{i=1}^h
        \frac{\dtable{w_i}{Y_i\uplus Z_i}}{\ptable{w_i}{Y_i\uplus Z_i}}\\
      = &
        \;\;\;\smashsum{A\subseteq E_{v,\leq h-1}^+;~B\subseteq \outedges_h;~A\cup B\neq\emptyset}\;
        \childprob[h-1]{A} \cdot \ptable{w_h}{Y_h \cup B} \cdot \sum_{i=1}^h
        \frac{\dtable{w_i}{Y_i\cup A_i \cup B_i}}{\ptable{w_i}{Y_i\cup A_i \cup B_i}}\\
      = &
        \;\;\;\smashsum{A\subseteq E_{v,\leq h-1}^+;~A\neq\emptyset}\;
        \childprob[h-1]{A} \cdot \ptable{w_h}{Y_h \cup \emptyset} \cdot \sum_{i=1}^h
        \frac{\dtable{w_i}{Y_i\cup A_i \cup \emptyset}}{\ptable{w_i}{Y_i\cup A_i \cup \emptyset}}\\
      & +
        \;\;\;\smashsum[l]{B\subseteq \outedges_h;~B\neq\emptyset}\;
        \;\;\;\smashsum[r]{A\subseteq E_{v,\leq h-1}^+}\;
        \childprob[h-1]{A} \cdot \ptable{w_h}{Y_h \cup B} \cdot \sum_{i=1}^h
        \frac{\dtable{w_i}{Y_i\cup A_i \cup B_i}}{\ptable{w_i}{Y_i\cup A_i \cup B_i}}\\
      = &
        \ptable{w_h}{Y_h} \cdot \;\;\;\smashsum{A\subseteq E_{v,\leq h-1}^+;~A\neq\emptyset}\;
        \childprob[h-1]{A} \cdot \left(\sum_{i=1}^{h-1}
        \frac{\dtable{w_i}{Y_i\cup A_i}}{\ptable{w_i}{Y_i\cup A_i}} + \frac{\dtable{w_h}{Y_h\cup A_h}}{\ptable{w_h}{Y_h\cup A_h}} \right)\\
      & +
        \;\;\;\smashsum{B\subseteq \outedges_h;~B\neq\emptyset}\;
        \ptable{w_h}{Y_h \cup B} \cdot 
        \;\;\;\smashsum{A\subseteq E_{v,\leq h-1}^+}\;
        \childprob[h-1]{A} \cdot \left( \sum_{i=1}^{h-1}
        \frac{\dtable{w_i}{Y_i\cup A_i \cup B_i}}{\ptable{w_i}{Y_i\cup A_i \cup B_i}} +
        \frac{\dtable{w_h}{Y_h\cup A_h \cup B_h}}{\ptable{w_h}{Y_h\cup A_h \cup B_h}} \right)\\
      \stackrel*= &
        \ptable{w_h}{Y_h} \cdot \;\;\;\smashsum{A\subseteq E_{v,\leq h-1}^+;~A\neq\emptyset}\;
        \childprob[h-1]{A} \cdot \left(\sum_{i=1}^{h-1}
        \frac{\dtable{w_i}{Y_i\cup A_i}}{\ptable{w_i}{Y_i\cup A_i}} + \frac{\dtable{w_h}{Y_h}}{\ptable{w_h}{Y_h}} \right)\\
      & +
        \;\;\;\smashsum{B\subseteq \outedges_h;~B\neq\emptyset}\;
        \ptable{w_h}{Y_h \cup B} \cdot 
        \;\;\;\smashsum{A\subseteq E_{v,\leq h-1}^+}\;
        \childprob[h-1]{A} \cdot \left( \sum_{i=1}^{h-1}
        \frac{\dtable{w_i}{Y_i\cup A_i}}{\ptable{w_i}{Y_i\cup A_i}} +
        \frac{\dtable{w_h}{Y_h\cup B}}{\ptable{w_h}{Y_h\cup B}} \right)\\
      = &
        \ptable{w_h}{Y_h} \cdot \;\;\;\smashsum{A\subseteq E_{v,\leq h-1}^+;~A\neq\emptyset}\;
        \childprob[h-1]{A} \cdot \sum_{i=1}^{h-1}
        \frac{\dtable{w_i}{Y_i\cup A_i}}{\ptable{w_i}{Y_i\cup A_i}}
      +
        \ptable{w_h}{Y_h} \cdot \;\;\;\smashsum{A\subseteq E_{v,\leq h-1}^+;~A\neq\emptyset}\;
        \childprob[h-1]{A} \frac{\dtable{w_h}{Y_h}}{\ptable{w_h}{Y_h}} \\
      & +
        \;\;\;\smashsum{B\subseteq \outedges_h;~B\neq\emptyset}\;
        \ptable{w_h}{Y_h \cup B} \cdot 
        \;\;\;\smashsum{A\subseteq E_{v,\leq h-1}^+}\;
        \childprob[h-1]{A} \cdot \sum_{i=1}^{h-1}
        \frac{\dtable{w_i}{Y_i\cup A_i}}{\ptable{w_i}{Y_i\cup A_i}}\\
      & +
        \;\;\;\smashsum{B\subseteq \outedges_h;~B\neq\emptyset}\;
        \ptable{w_h}{Y_h \cup B} \cdot 
        \;\;\;\smashsum{A\subseteq E_{v,\leq h-1}^+}\;
        \childprob[h-1]{A} \cdot
        \frac{\dtable{w_h}{Y_h\cup B}}{\ptable{w_h}{Y_h\cup B}}\\
      = &
        \ptable{w_h}{Y_h} \cdot C\dtable{v}{Y,h-1}
      +
        \dtable{w_h}{Y_h} \cdot C\ptable{v}{Y,h-1}\\
      & +
        \;\;\;\smashsum{B\subseteq \outedges_h;~B\neq\emptyset}\;
        \ptable{w_h}{Y_h \cup B} \cdot 
        d_{h-1}^{C+Q}
      +
        \;\;\;\smashsum{B\subseteq \outedges_h;~B\neq\emptyset}\;
        \dtable{w_h}{Y_h \cup B} \cdot 
        p_{h-1}^{C+Q}
    \end{align*}
    In Equality~*, we observe that~$A\subseteq E_{v,\leq h-1}^+$ and~$B\subseteq \outedges_h$ which implies~$A_h=\emptyset$ and~$B_i=\emptyset$ for all~$i<h$.
    This shows that tables~$C$ are computed correctly.
  \end{proof}

  We can now compute $\ptable{v}{Y}$ and $\dtable{v}{Y}$ using \eqref{eq:sw-internal-p} and \eqref{eq:sw-internal-d}
  by replacing the sums over all $Z\subseteq \outedges$ by lookups into $C\ptable{v}{Y,\Delta}$ and $C\dtable{v}{Y,\Delta}$.
  More formally, $\ptable{v}{Y} = p_I(uv) \cdot C\ptable{v}{Y,h}$ and~$\dtable{v}{Y} = p_I(uv) \cdot C\dtable{v}{Y,h} + \w(uv) \cdot \ptable{v}{Y}$.\todos{Please check me!}

  To prove \cref{thm:FPT-sw}, it remains to formally prove that $\APD(X_N) = \dtable{\rho}{\emptyset}$, as well as the running time.
  For this, we need some additional notation.
  We say that a vertex~$r$ of $N$ is \emph{switched} by a switching~$\sigma$ if $\sigma$ contains only one incoming edge of~$r$ and
  we let~$R(\sigma)$ be the set of such vertices.
  We define the combining operation~$\oplus$ on switchings, such that $\sigma_1 \oplus \sigma_2 := (V(N), E(\sigma_1)\cap E(\sigma_2))$,
  that is, all vertices that are switched by~$\sigma_1$ or by~$\sigma_2$ are switched accordingly by $\sigma_1\oplus\sigma_2$.
  If~$R(\sigma_1)\cap R(\sigma_2)=\emptyset$, then 
  \begin{equation}
    \label{eq:merged-prob}
    P(\sigma_1\oplus\sigma_2) 
    = \;\smashprod{\substack{ur\in E(\sigma_1\oplus \sigma_2)\\ r\in R(\sigma_1\oplus\sigma_2)}}\; p_I(ur)
    = \;\smashprod{\substack{ur\in E(\sigma_1)\cap E(\sigma_2)\\ r\in R(\sigma_1) \uplus R(\sigma_2)}}\; p_I(ur)
    = \smashprod{\substack{ur\in E(\sigma_1)\\ r\in R(\sigma_1)}} p_I(ur) \cdot \smashprod{\substack{ur\in E(\sigma_2)\\ r\in R(\sigma_2)}} p_I(ur)
    = P(\sigma_1)\cdot P(\sigma_2).
  \end{equation}
  For sets of switchings~$\mathfrak{S}_1,\mathfrak{S}_2$,
  we define $\mathfrak{S}_1\oplus\mathfrak{S}_2:=\{\sigma_1\oplus\sigma_2\mid\sigma_1\in\mathfrak{S}_1 \land \sigma_2\in\mathfrak{S}_2\}$
  and we note that, if $R(\sigma_1)\cap R(\sigma_2)=\emptyset$ for all $\sigma_1\in\mathfrak{S}_1$ and $\sigma_2\in\mathfrak{S}_2$, then
  \begin{equation}
    P(\mathfrak{S}_1\oplus\mathfrak{S}_2)
    = 
      \smashsum{\sigma_1\in\mathfrak{S}_1,~\sigma_2\in\mathfrak{S}_2} P(\sigma_1\oplus\sigma_2)
    \stackrel{\eqref{eq:merged-prob}}{=}
      \smashsum{\sigma_1\in\mathfrak{S}_1,~\sigma_2\in\mathfrak{S}_2} P(\sigma_1)\cdot P(\sigma_2)
    = P(\mathfrak{S}_1)\cdot P(\mathfrak{S}_2)
    \label{eq:merged-prob-sets}
  \end{equation}

\newcommand{\proofthmswfptcorrectness}{
  In the proof, we use $e\leq_N \GW(v)$ as a shorthand to mean that $N$ contains a path starting with an edge in $\GW(v)$ and ending with the edge $e$.
  First, we show the correctness of the dynamic programming by induction on the height of~$v$ in $\Gamma$.
  To this end, we prove for all~$Y\subseteq \GW(v)$ that the entries~$\ptable{v}{Y}$ and $\dtable{v}{Y}$
  contain the probability and diversity score (restricted to edges below~$\GW(v)$) of the set of all $(v,Y)$-compatible switchings in $\swis(N)$.

  If $v$ is a leaf in $\Gamma$ and, thus, in $N$, then $\inedges=\{uv\}$ for some vertex~$u$ of $N$.
  Thus, $\haspath[\sigma]{uv}{X_N}$ for all switchings~$\sigma\in\swis(N)$,
  so all switchings are $(v,\{uv\})$-compatible, but not $(v,\emptyset)$-compatible.
  As the diversity score of any such switching, restricted to~$uv$ is $\w(uv)$, 
  the correctness of \eqref{eq:sw-base} follows.

  For the rest of the proof, suppose that~$v$ has children~$w_1,w_2,\ldots,w_\Delta$ in $\Gamma$ and
  the tables store the correct value for all proper descendants of~$v$ in $\Gamma$.
  For all $u\in V(N)$ and $Y\subseteq\GW(u)$,
  let~$\compsw^u_Y$ denote the set of $(u,Y)$-compatible partial switchings of~$N$
  that switch only reticulations in the subtree of~$\Gamma$ rooted at~$u$.
  We consider three cases for $\inedges\cap Y$.

  \textbf{Case 1}: $|\inedges\cap Y|\geq 2$.
  Then, no switching is $(v,Y)$-compatible, so their probability is $0 = \ptable{v}{Y}$ and their combined score is $0 = \dtable{v}{Y}$.

  \textbf{Case 2}: $\inedges\cap Y=\emptyset$.
	Then, $\nhaspath[\sigma]{uv}{X_N}$ for all $(v,Y)$-compatible switchings.
	We start by showing $\compsw^v_Y = \mathfrak{S}_v \oplus \bigoplus_{i=1}^\Delta \compsw^{w_i}_{Y_i}$,
  where~$\mathfrak{S}_v$ is the set of partial switchings in which only~$v$ is switched.
	If~$v$ is a tree-vertex, then~$\mathfrak{S}_v$ only contains~$N$.
	\begin{lemma}\label{lem:SvY}
    Let~$\inedges\cap Y=\emptyset$.
    Then, $\compsw^v_Y = \mathfrak{S}_v \oplus \bigoplus_{i=1}^\Delta \compsw^{w_i}_{Y_i}$
	\end{lemma}
	\begin{proof}
		``$\subseteq$'':
    Let~$\sigma$ be a partial switching in~$\compsw^v_Y$.
    Then, all reticulations in the subtree of~$\Gamma$ rooted at~$v$ are switched in~$\sigma$ and,
    therefore, all reticulations in all subtrees of~$\Gamma$ rooted at any~$w_i$.
		Let~$\sigma_i$ be the partial switching of only the reticulations in the subtree of~$\Gamma$ rooted at any~$w_i$.
		Further, let~$\sigma_v$ be the switching that only switches~$v$ according to how~$v$ is switched in $\sigma$
    (if~$v$ is a tree-vertex, then~$\sigma_v$ is~$N$).
		Clearly, $\sigma = \sigma_1 \oplus \dots \oplus \sigma_\Delta \oplus \sigma_v$.
		It remains to show that~$\sigma_i \in \compsw^{w_i}_{Y_i}$, for each~$i\in [\Delta]$,
    that is, $\sigma_i$ is $(w_i,Y_i)$-compatible and switches only the reticulations below~$w_i$ in $\Gamma$.
    As the second condition is satisfied by the construction of $\sigma_i$, we only show the first.
    To this end, first consider some~$uw\in Y_i$ and assume towards a contradiction that~$\nhaspath[\sigma_i]{uw}{X_N}$,
    implying $\nhaspath[\sigma]{uw}{X_N}$ since $\sigma$ is a subgraph of~$\sigma_i$.
    But then, $uw\notin Y$ since $\sigma$ is $(v,Y)$-compatible, contradicting~$uw\in Y_i$.
    Second, consider some arc~$uw\in\GW(w_i)$ with $\haspath[\sigma_i]{uw}{X_N}$ and assume towards a contradiction that $uw\notin Y_i$.
    Then, $uw\notin\GW(v)$ as, otherwise, $\haspath[\sigma]{uw}{X_N}$ and $uw\in Y$ since $\sigma$ is $(v,Y)$-compatible, contradicting $uw\notin Y_i$.
    But the only arcs in $\GW(w_i)\setminus\GW(v)$ are outgoing from~$v$, so $u=v$ and,
    since $\inedges\cap Y=\emptyset$, we know that $v$ does not have a path to a leaf in $\sigma$,
    in particular, $\nhaspath[\sigma]{uw}{X_N}$, implying $\nhaspath[\sigma_i]{uw}{X_N}$ since $\haspath[\Gamma]{w_i}{w}$.
		
    ``$\supseteq$'':
		Let $\sigma_i \in \compsw^{w_i}_{Y_i}$ for each~$i \in [\Delta]$ and let~$\sigma_v \in \mathfrak{S}_v$.
		Let~$\sigma := \sigma_1 \oplus \dots \oplus \sigma_\Delta \oplus \sigma_v$ and
    observe that~$\sigma$ switches all reticulations in the subtree of~$\Gamma$ rooted at~$v$.
    We show that $\sigma$ is $(v,Y)$-compatible.
    First, consider some~$uw \in Y$.
    Since~$\inedges\cap Y=\emptyset$, there is some~$w_i$ with~$uw \in \GW(w_i)$, so~$uw\in Y_i$.
    Then, as~$w_i$ is~$(w_i,Y_i)$-compatible, we know~$\haspath[\sigma]{uw}{X_N}$ and $\haspath[\sigma_i]{uw}{X_N}$.
    Second, consider some~$uw \in \GW(v) \setminus Y$; in particular, $uw\notin Y_i$ for any~$i$, implying $\nhaspath[\sigma_i]{uw}{ X_N}$.
    If~$w \neq v$, then~$uw \in \GW(w_i)$ for some~$i$ so~$\nhaspath[\sigma]{uw}{X_N}$ and $\nhaspath[\sigma_i]{uw}{X_N}$.
    If~$w=v$, then observe that no~$Y_i$ contains an outgoing arc of~$v$.
    In consequence, there is no path from~$v$ to any leaf in any~$\sigma_i$ and, thus, $\nhaspath[\sigma]{uv}{X_N}$.
	\end{proof}

	Thus, the probability of drawing a $(v,Y)$-compatible switching is
	\begin{align*}
		P(\compsw^v_Y)
		& = P\left(\mathfrak{S}_v \oplus \bigoplus_{i=1}^{\Delta} \compsw^{w_i}_{Y_i}\right)
		\stackrel{\eqref{eq:merged-prob-sets}}{=} \prod_{i=1}^\Delta \; P\left(\compsw^{w_i}_{Y_i}\right) \cdot P(\mathfrak{S}_v)
		= \prod_{i=1}^\Delta \; P\left(\compsw^{w_i}_{Y_i}\right)\\
		&\stackrel{\text{IH}}{=} \prod_{i=1}^\Delta \ptable{w_i}{Y_i}
		= \childprob[\Delta]{\emptyset}
		\stackrel{\eqref{eq:sw-internal-pd-empty}}{=} \ptable{v}{Y},
	\end{align*}
	where we use that~$P(\mathfrak{S}_v) = \sum_{uv \in \inedges} p_I(uv) = 1$.
	Now, their combined \APDlong (restricted to edges below $\GW(v)$) is
	\begin{align*}
    \smashsum[r]{\sigma\in\compsw^v_Y}\; P(\sigma)\cdot \;\;\smashsum{\substack{e \beloeq \GW(v)\\ \haspath[\sigma]{e}{X_N}}}\; \w(e)
    & \stackrel{\text{\cref{lem:SvY}}}{=}
      \;\smashsum{\sigma\in \mathfrak{S}_v \oplus \bigoplus_{i=1}^\Delta \compsw^{w_i}_{Y_i}}\;\; P(\sigma) \;\cdot
		    \;\;\smashsum{\substack{e \beloeq \GW(v)\\ \haspath[\sigma]{e}{X_N}}}\; \w(e)\\
		\intertext{and, considering the edges below each $\GW(w_j)$ separately}
		& =
      \sum_{j=1}^\Delta\;
        \smashsum[r]{\sigma\in \mathfrak{S}_v \oplus \bigoplus_{i=1}^\Delta\; \compsw^{w_i}_{Y_i}} \; P(\sigma) \;\;\cdot
		      \;\;\;\smashsum{\substack{e \beloeq \GW(w_j)\\ \haspath[\sigma]{e}{X_N}}}\; \w(e)\\
		& =
      \sum_{j=1}^\Delta\;
        \sum_{\sigma_1\in \compsw^{w_1}_{Y_1}} \cdots \sum_{\sigma_\Delta\in\compsw^{w_\Delta}_{Y_\Delta}} \sum_{\sigma_v \in \mathfrak{S}_v} P(\sigma_1 \oplus \cdots \oplus \sigma_\Delta \oplus \sigma_v) \cdot
		      \;\smashsum{\substack{e \beloeq \GW(w_j)\\ \haspath[\sigma_j]{e}{X_N}}}\; \w(e)\\
		& \stackrel{\eqref{eq:merged-prob}}{=}
      \sum_{j=1}^\Delta\;
        \sum_{\sigma_1\in \compsw^{w_1}_{Y_1}} \cdots \sum_{\sigma_\Delta\in\compsw^{w_\Delta}_{Y_\Delta}} \sum_{\sigma_v \in \mathfrak{S}_v} P(\sigma_1) \cdots P(\sigma_\Delta) \cdot P(\sigma_v) \cdot
		    \;\smashsum{\substack{e \beloeq \GW(w_j)\\ \haspath[\sigma_j]{e}{X_N}}}\; \w(e)\\
		\intertext{We can take the terms that do not depend on~$j$ before the sum with~$\w$.}
		& =
      \sum_{j=1}^\Delta\;
		    \bigg(
          \underbrace{\sum_{\sigma_1\in\compsw^{w_1}_{Y_1}} \cdots \sum_{\sigma_\Delta\in\compsw^{w_\Delta}_{Y_\Delta}}}_{\text{not for~$j$}} \sum_{\sigma_v \in \mathfrak{S}_v}
		      \underbrace{P(\sigma_1) \cdots P(\sigma_\Delta)}_{\text{not for~$j$}} \cdot P(\sigma_v)
        \bigg)
        \cdot
        \smashsum{\sigma_j\in \compsw^{w_j}_{Y_j}} P(\sigma_j) \cdot
          \;\smashsum{\substack{e \beloeq \GW(w_j)\\ \haspath[\sigma_j]{e}{X_N}}}\; \w(e)\\
    \intertext{Then, by distributivity,}
		& =
      \sum_{j=1}^\Delta\;
		    \bigg(
          \bigg(\sum_{\sigma_v\in \mathfrak{S}_v}P(\sigma_v)\bigg)
          \cdot
          \prod_{\stackrel{i=1}{i\neq j}}^\Delta \sum_{\sigma_i\in\compsw^{w_i}_{Y_i}} P(\sigma_i)
        \bigg)
        \cdot
        \smashsum{\sigma_j\in \compsw^{w_j}_{Y_j}} P(\sigma_j)\cdot
          \;\smashsum{\substack{e \beloeq \GW(w_j)\\ \haspath[\sigma_j]{e}{X_N}}}\; \w(e)\\
    & =
      \sum_{j=1}^\Delta\;
		    \bigg(
          P(\mathfrak{S}_v)
          \cdot
          \prod_{\stackrel{i=1}{i\neq j}}^\Delta P(\compsw^{w_i}_{Y_i})
        \bigg)
        \cdot
        \smashsum{\sigma_j\in \compsw^{w_j}_{Y_j}} P(\sigma_j) \cdot
          \;\smashsum{\substack{e \beloeq \GW(w_j)\\ \haspath[\sigma_j]{e}{X_N}}}\; \w(e)\\
    \intertext{With~\Cref{lem:SvY} and Equation~\eqref{eq:merged-prob-sets} we know~$P(\mathfrak{S}_v)
    	\cdot
    	\prod_{i=1}^\Delta P(\compsw^{w_i}_{Y_i})
    	= P(\compsw^{v}_{Y})$. So,}
		& =
      \sum_{j=1}^\Delta\;
		    \frac{P(\compsw^{v}_{Y})}{P(\compsw^{w_j}_{Y_j})}
        \cdot
		    \underbrace{
			    \sum_{\sigma_j\in \compsw^{w_j}_{Y_j}} P(\sigma_j)
          \cdot
			    \;\smashsum{\substack{e \beloeq \GW(w_j)\\ \haspath[\sigma_j]{e}{X_N}}}\; \w(e)
        }_{\stackrel{\text{IH}}{=}\,\dtable{w_j}{Y_j}}\\
		& \stackrel{\text{IH}}{=}
		  P(\compsw^{v}_{Y})
      \cdot
      \sum_{j\leq\Delta}
        \frac{\dtable{w_j}{Y_j}}{P\left(\compsw^{w_j}_{Y_j}\right)}
		= \childprob{\emptyset} \cdot \sum_{j\leq\Delta} \frac{\dtable{w_j}{Y_j}}{\ptable{w_j}{Y_j}}
		\stackrel{\eqref{eq:sw-internal-pd-empty}}{=} \dtable{v}{Y}
	\end{align*}

	\textbf{Case 3}: $\inedges\cap Y=\{uv\}$ for some~$u^*\in V(N)$.
	Then, $\haspath[\sigma]{uv}{X_N}$ for all $(v,Y)$-compatible switchings.
	We start by showing an adaptation of \Cref{lem:SvY}.
	In the remainder of the section,
  let~$\sigma_v$ be the partial switching resulting from removing all incoming edges of~$v$, except for~$uv$
  (that is, in~$\sigma_v$ only~$v$ is switched; if~$v$ is a tree-vertex, then~$\sigma_v = N$),
  and abbreviate~$YZ_i := (Y \cup Z) \cap \GW(w_i)$ for all~$Z \subseteq \outedges[v]$.
	\begin{lemma}\label{lem:SvYZ}
		Let~$|\inedges\cap Y| = 1$.
    Then, $\compsw^v_Y = \{\sigma_v\}\oplus \biguplus_{{Z \subseteq \outedges[v]},\,{Z \neq \emptyset}} \bigoplus_{i=1}^\Delta \compsw^{w_i}_{YZ_i}$.
	\end{lemma}
	\begin{proof}
		We assume in this proof, without loss of generality, that~$\inedges\cap Y=\{uv\}$.
		
		``$\subseteq$'':
		Let~$\sigma \in \compsw^v_Y$.
		Let~$Z \subseteq \outedges[v]$ be the set of arcs~$vw$ with~$\haspath[\sigma]{vw}{X_N}$.
		Since~$\haspath[\sigma]{uv}{X_N}$, this set is non-empty.
		Define~$\sigma_i$ to be the partial switching which switches reticulations in the subtree of~$\Gamma$ rooted at~$w_i$ as~$\sigma$ and does not switch any other reticulation.
		Thus, $\sigma = \sigma_v \oplus \sigma_1 \oplus \dots \oplus \sigma_\Delta$.
		It remains to show that~$\sigma_i \in \compsw^{w_i}_{YZ_i}$, for each~$i \in [\Delta]$.
		Fix~$i$.
		By definition, $\sigma_i$ switches the right reticulations.
		$\haspath[\sigma]{vw}{X_N}$, for each~$vw \in Z \supseteq Z_i$, by definition of~$Z$ and $\haspath[\sigma]{e}{X_N}$, for each arc~$e \in Y \supseteq Y_i$.
		Since~$\sigma_i$ switches all vertices in the subtree of~$\Gamma$ rooted at~$w_i$ as~$\sigma$, these properties also hold for~$\sigma_i$ and we are done.
		
		``$\supseteq$'':
		Fix any non-empty~$Z \subseteq \outedges[v]$.
		Let~$\sigma_i \in \compsw^{w_i}_{YZ_i}$, for each~$i\in [\Delta]$.
		Define~$\sigma := \sigma_v \oplus \sigma_1 \oplus \dots \oplus \sigma_\Delta$.
		We show~$\sigma \in \compsw^v_Y$.
		In $\sigma$, all reticulations in the subtree of $\Gamma$ rooted at~$v$ are switched.
		For each arc~$a\in Y\setminus \inedges[v]$, there is an~$i \in [\Delta]$, such that~$a \in Y_i$.
		Then, $\haspath[\sigma]{a}{X_N}$ and $\haspath[\sigma_i]{a}{X_N}$.
		Since~$Z$ is non-empty, there is an arc~$vw \in Z_i$, for some~$i \in [\Delta]$, such that~$\haspath[\sigma]{vw}{X_N}$ and $\haspath[\sigma_i]{vw}{X_N}$.
		We conclude that~$\haspath[\sigma]{uv}{X_N}$ and the claim follows.
	\end{proof}
	
  By \cref{lem:SvYZ}, the probability of drawing a $(v,Y)$-compatible switching is
	\begin{align*}
		P(\compsw^v_Y)
    & \stackrel{\text{\cref{lem:SvYZ}}}{=}
      P\bigg(
        \{\sigma_v\}\oplus
        \;\smashoperator[l]{\biguplus_{{Z \subseteq \outedges[v]},\,{Z \neq \emptyset}}}\;
        \bigoplus_{i=1}^\Delta \compsw^{w_i}_{YZ_i}
      \bigg)
		\stackrel{\eqref{eq:merged-prob-sets}}{=} 
      P(\sigma_v)\cdot
      \;\smashsum[l]{{Z \subseteq \outedges[v]},\,{Z \neq \emptyset}}\;
        \prod_{i=1}^\Delta P\left(\compsw^{w_i}_{YZ_i}\right)\\
		& =
      p_I(uv) \cdot \;\smashsum[l]{{Z \subseteq \outedges[v]},\,{Z \neq \emptyset}}\; \prod_{i=1}^\Delta P\left(\compsw^{w_i}_{YZ_i}\right)\\
		& \stackrel{\text{IH}}{=}
      p_I(uv) \cdot \;\smashsum[l]{{Z \subseteq \outedges[v]},\,{Z \neq \emptyset}}\prod_{i=1}^\Delta \ptable{w_i}{YZ_i}
    =
      p_I(uv) \cdot \;\smashsum[l]{{Z \subseteq \outedges[v]},\,{Z \neq \emptyset}} \childprob[\Delta]{Z}
		\stackrel{\eqref{eq:sw-internal-p}}{=} \ptable{v}{Y}.
	\end{align*}
	
	\noindent
	Now, their combined \APDlong (restricted to edges below $\GW(v)$) is
	\begin{align*}
		&
      \smashsum[r]{\sigma\in\compsw^v_Y}\; P(\sigma)\cdot \;\;\smashsum{\substack{e \beloeq \GW(v)\\ \haspath[\sigma]{e}{X_N}}}\; \w(e)\\
		\intertext{and, considering the edge $uv$ as well as the edges below each $\GW(w_j)$ separately}
    = &
      \;\smashsum{\sigma\in\compsw^v_Y}\; P(\sigma)\cdot
      \bigg(
        \w(uv) + \sum_{j=1}^{\Delta} \smashsum[r]{\substack{e \beloeq \GW(w_j)\\ \haspath[\sigma]{e}{X_N}}}\; \w(e)
      \bigg)\\
    = &
      P(\compsw^v_Y)\cdot\w(uv) +
      \;\smashsum{\sigma\in\compsw^v_Y}\; P(\sigma)\cdot
      \sum_{j=1}^{\Delta} 
        \smashsum[r]{\substack{e \beloeq \GW(w_j)\\ \haspath[\sigma]{e}{X_N}}}\; \w(e)\\
		\stackrel{\text{\cref{lem:SvYZ}}}{=} &
      P(\compsw^v_Y)\cdot\w(uv) +
		  \sum_{\substack{Z \subseteq \outedges[v] \\ Z \neq \emptyset}}\;
        \smashsum[r]{\sigma\in \bigoplus_{i=1}^\Delta \compsw^{w_i}_{YZ_i}}\;\; P(\sigma_v \oplus \sigma) \cdot
          \sum_{j=1}^{\Delta} 
          \smashsum[r]{\substack{e \beloeq \GW(w_j)\\ \haspath[\sigma_v \oplus \sigma]{e}{X_N}}}\; \w(e)\\
    \intertext{and, considering the parts below each $w_j$ of each switching~$\sigma$ individually}
		= &
    P(\compsw^v_Y)\cdot\w(uv) +
		\smashsum{\substack{Z \subseteq \outedges[v] \\ Z \neq \emptyset}}\;
		  \sum_{j=1}^\Delta\;
        \smashsum[r]{\sigma_1\in \compsw^{w_1}_{YZ_1}} \;\;\cdots\;\; \smashsum{\sigma_\Delta\in\compsw^{w_\Delta}_{YZ_\Delta}}
          P(\sigma_v \oplus \sigma_1 \oplus \cdots \oplus \sigma_\Delta) \cdot
		      \;\smashsum{\substack{e \beloeq \GW(w_j)\\ \haspath[\sigma_j]{e}{X_N}}}\; \w(e)\\
		\stackrel{\eqref{eq:merged-prob}}{=} &
      P(\compsw^v_Y)\cdot\w(uv) +
      \;\smashsum[l]{\substack{Z \subseteq \outedges[v] \\ Z \neq \emptyset}}\;
		    \sum_{j=1}^\Delta\;
        \smashsum[r]{\sigma_1\in \compsw^{w_1}_{YZ_1}} \;\;\cdots\;\; \smashsum{\sigma_\Delta\in\compsw^{w_\Delta}_{YZ_\Delta}}
          p_I(uv) \cdot P(\sigma_1) \cdots P(\sigma_\Delta) \cdot
		      \;\smashsum{\substack{e \beloeq \GW(w_j)\\ \haspath[\sigma_j]{e}{X_N}}}\; \w(e)\\
    \intertext{with $P(\compsw^v_Y) =\ptable{v}{Y}$ (see above) and similarly to Case~2, we get}
		= &
    \ptable{v}{Y}\cdot\w(uv) +
		p_I(uv)\cdot
    \;\smashsum[l]{\substack{Z \subseteq \outedges[v] \\ Z \neq \emptyset}}\;
		  \sum_{j=1}^\Delta\;
		    \bigg(
		      \prod_{\substack{i=1\\ i\neq j}}^\Delta P(\compsw^{w_i}_{YZ_i})
		    \bigg)
		    \cdot
		    \smashsum{\sigma_j\in \compsw^{w_j}_{YZ_j}} P(\sigma_j) \cdot
		      \smashsum{\substack{e \beloeq \GW(w_j)\\ \haspath[\sigma_j]{e}{X_N}}}\; \w(e)\\
		= &
    \ptable{v}{Y}\cdot\w(uv) +
		p_I(uv)\cdot
    \;\smashsum[l]{\substack{Z \subseteq \outedges[v] \\ Z \neq \emptyset}}\;
		  \sum_{j=1}^\Delta\;
		    \bigg(
		      \underbrace{
            \prod_{i=1}^\Delta
            \underbrace{
              P(\compsw^{w_i}_{YZ_i})
            }_{\stackrel{\text{IH}}{=}\, \ptable{w_i}{YZ_i}}
          }_{= \childprob{Z}}
        \bigg)
          /
          \underbrace{
            P(\compsw^{w_j}_{YZ_j})
          }_{\stackrel{\text{IH}}{=}\, \ptable{w_j}{YZ_j}}
		    \cdot
    		\underbrace{
          \smashsum{\sigma_j\in \compsw^{w_j}_{YZ_j}} P(\sigma_j) \cdot
            \;\smashsum{\substack{e \beloeq \GW(w_j)\\ \haspath[\sigma_j]{e}{X_N}}}\; \w(e)
        }_{\stackrel{\text{IH}}{=}\, \dtable{w_j}{YZ_j}}\\
		= &
		\ptable{v}{Y} \cdot \w(u v) +
    p_I(uv)\cdot
		\;\smashsum{\substack{Z \subseteq \outedges[v] \\ Z \neq \emptyset}}\;
    \childprob{Z}\cdot
		\sum_{j=1}^\Delta\;
    \frac{\dtable{w_j}{YZ_j}}{\ptable{w_j}{YZ_j}}
		\stackrel{\eqref{eq:sw-internal-d}}{=}
		\dtable{v}{Y}.
	\end{align*}

  \noindent
	Then, by the correctness of the computation,
	$\APD[N](X_N) = \dtable{\rho}{\emptyset}$.
}

\newcommand{\proofthmswfpt}[1]{
	\begin{proof}[Proof of \cref{thm:FPT-sw}]
	#1
    
    \looseness=-1
    Observe that there are $\Oh(2^{s} \cdot |E(\Gamma)|)=\Oh(2^{s}\cdot n)$ entries in all tables, together, for a width-$s$ tree-extension.
    Most table entries are computed by looking up a constant number of table entries.
    We recall that the calculations in \eqref{eq:sw-internal-p} and \eqref{eq:sw-internal-d}
    are replaced with lookups into the auxiliary tables and
    the calculations of the $Q$-table
    are defined in \Cref{lem:sw-auxQ-correct,lem:sw-auxC-correct}.
    It only remains to show that the
    entries of the $C$-table, as presented in \Cref{lem:sw-auxC-correct},
    can be computed in $\Oh(2^{s}\cdot n)$~time.
    
    To this end,
    fix a vertex~$v$ with children $w_1,\ldots,w_\Delta$ in $\Gamma$,
    fix any~$h\leq\Delta$,
    and note that there are $|E(\Gamma)|\in\Oh(n)$ such pairs~$(v,h)$.
    Further, consider the pairs~$(Y_h,Z_h)$ with $Y_h\subseteq \GW(v)\cap\GW(w_h)$ and $Z_h\subseteq\outedges[v]\cap\GW(w_h)$.
    As $\GW(v)\cap\outedges[v]=\emptyset$, there are at most $2^{|\GW(v)\cap\GW(w_h)|}\cdot 2^{|\outedges[v]\cap\GW(w_h)|} \leq 2^{|\GW(w_h)|} \leq 2^{\sw}$ such pairs.
    Thus, the two values
    $\smashsum{Z\subseteq \outedges_h;~Z\neq\emptyset}\;\; p^{C+Q}_{h-1} \cdot \ptable{w_h}{Y^*\cup Z}$ and
    $\smashsum{Z\subseteq \outedges_h;~Z\neq\emptyset}\;\; \dtable{v}{Y^*\cup Z}\cdot p^{C+Q}_{h-1} + d^{C+Q}_{h-1} \cdot \ptable{w_h}{Y^*\cup Z}$,
    can be computed in~$\Oh(2^{s})$ time.
    Consequently, taking these values as constants after the calculation, all values in \Cref{lem:sw-auxC-correct} are computed in the desired time.
    This proves the overall running time of~$\Oh(2^{s} n)$.
  \end{proof}
}
\proofthmswfpt{The correctness proof is deferred to the preprint and journal version.}

\thmtoappendix{thm:FPT-sw}{
	\thmswfpt{}
}{
	\proofthmswfpt{\proofthmswfptcorrectness
	
	}
}

	\ifJournal
    \todo[inline]{MW: this needs to be removed or redone completely before journal-submission}
	\begin{theorem}\label{thm:FPT-nsw}
		For every network~$N$ on~$X_N$,
		\APDX can be computed in $\Oh(????)$ time,
		when a tree extension~$T$ of~$N$ with node scanwidth~$\nsw$ is given.
	\end{theorem}
	\begin{proof}
		Let a tree-extension~$T$ of~$N$ with node scanwidth~$\nsw$ be given.
		Let~$T_{\le v}$ be the subtree of~$T$ rooted at~$v$.
    Recall that~$B_v \subseteq V(N)$, for each~$v\in V(N)$, is the set of vertices that are a parent of a vertex in~$T_{\le v}$, but are not in~$T_{\le v}$.
		
		\proofpara{Algorithm}
		We define a dynamic programming algorithm with table~$\DP[v,Y]$, for dimensions~$v\in V(N)$ and~$Y \subseteq B_v$.
		We say that a display-tree~$\sigma$ of~$N$ is~$Y$-compatible\todo{overwriting the definition in \Cref{thm:FPT-sw}} if
		\begin{itemize}
			\item each vertex in~$Y$ has, in~$\sigma$, a directed path to a leaf in~$X_N$,
			\item Vertices that are not in~$Y$ do not have a directed path to a leaf in~$X_N$, in~$\sigma$.
		\end{itemize}
		We store in~$\DP[v,Y]$ a tuple~$[p;d]$, indicating that the proportion of~$Y$-compatible switching-trees~$\sigma$ is~$p \in [0,1]$ and that the overall average diversity of these trees on edges incoming at~$T_{\le v}$ is~$d \in \mathbb{R}_{\ge 0}$.
		If there is no such switching tree, we store~$[0,0]$.
		We briefly write~$\DP[v,Y,1]=p$ and~$\DP[v,Y,2]=d$.
		
		
		Each leaf~$x\in X_N$ has exactly one incoming edge. Observe~$\GW(x) = \parents(x)$.
		We store
		\begin{equation}
			\label{eq:nsw-base}
			\DP[x,Y] =~
			\begin{cases}
				[1; \w( v_x x )] & \text{if~$B_x = Y = \{v_x\}$}\\
				[0; 0] & \text{else if~$Y = \emptyset$}
			\end{cases}
		\end{equation}
		
		Now, let~$v$ be an internal vertex of~$N$, either a tree-vertex or a reticulation, and let~$w_1,\dots,w_\Delta$ be the children of~$v$ in the tree-extension~$T$.
		(Note that~$v$ can have more children in~$N$.
		We may assume that~$T$ is chosen in a manner that $v$ does not have fewer children in~$N$.)
		We use an auxiliary table~$\DP'$, where in~$\DP'[j,v,Y]$ we only consider the first~$j\leq\Delta$ children.
		We store
		\begin{equation}
			\label{eq:nsw-sub-1}
			\DP'[1,v,Y] =~
			\begin{cases}
				\DP[w_1,Y] & \text{if~$Y \subseteq \GW(w_1)$}\\
				[0;0] & \text{else}\\
			\end{cases}
		\end{equation}
		
		Let~$\mathcal{H}(Y)$ for some set~$Y$ be the set of tuples~$(Y_1,Y_2)$ with~$Y_i\subseteq Y$ for~$i\in\{1,2\}$\todos{I think we already have this condition satisfied with the next condition. I would keep it, however, because of reader-friendliness.} and~$Y_1 \cup Y_2 = Y$.
		(We do not require~$Y_1$ and~$Y_2$ to be disjoint.)
		For further values, we use the recursion
		\begin{align}
			\label{eq:nsw-sub-p-j}
			\DP'[j+1,v,Y,1]
			=~& \sum_{(Y_1,Y_2) \in \mathcal{H}(Y)}
			\DP'[j,v,Y_1,1] \cdot \DP[w_{j+1},v,Y_2,1]\\
			\label{eq:nsw-sub-d-j}
			\DP'[j+1,v,Y,2]
			=~& \sum_{(Y_1,Y_2) \in \mathcal{H}(Y)}
			\DP'[j,v,Y_1,1] \cdot \DP[w_{j+1},v,Y_2,2] +\\
			& ~~~~~~~~~~~~~~~~ \DP'[j,v,Y_1,2] \cdot \DP[w_{j+1},v,Y_2,1]
		\end{align}
		\todosi{Interesting observation for the running time analysis: Each vertex in~$\GW(v)\cup\{v\}$ can be one of the following four options:\\
		- Not in $Y$\\
		- In $Y_1$, but not in~$Y_2$\\
		- In $Y_2$, but not in~$Y_1$\\
		- In $Y_1$ and in~$Y_2$\\
    Thus computing all values in the DP should be feasible in~$\Oh(4^{\nsw} \cdot n)$~time.}
		
		Finally, we store
		\begin{equation}
			\label{eq:nsw-final-p}
			&\DP[v,Y,1]\\
			=~&	\DP[\Delta,v,Y,1] +\\ 
			&	\sum_{y \in Y \cap \parents(v)}
				p_I(yv) \cdot ( \DP[\Delta,v,Y \cup \{v\},1] + \DP[\Delta,v,Y \setminus \{y\} \cup \{v\},1] )\\
			\label{eq:nsw-final-d}
			&\DP[v,Y,2]\\
			=~&	\DP[\Delta,v,Y,2] +\\
			&	\sum_{y \in Y \cap \parents(v)}
				p_I(yv) \cdot ( \DP[\Delta,v,Y \cup \{v\},2] + \DP[\Delta,v,Y \setminus \{y\} \cup \{v\},2]\\
			&	~~~~~~~~~~~~~ + \w(yv) \cdot ( \DP[\Delta,v,Y \cup \{v\},1] + \DP[\Delta,v,Y \setminus \{y\} \cup \{v\},1] ) )
		\end{equation}

		We return \yes, if $\DP[\rho,\emptyset,2] \ge D$, where~$\rho$ is the root of~$N$---and therefore also the root of the tree-extension~$T$.

		\proofpara{Correctness}

		\proofpara{Running time}

	\end{proof}
	\fi

	\section{\APDX on Almost Reticulation-Visible Networks}\label{sec:rv}
  Recall that a network is \emph{reticulation-visible} if all its reticulations are visible,
  e.g.~for each reticulation~$r$, there is a leaf~$\ell$ such that $r$ is on all root-$\ell$-paths in $N$.
  
  We show first that,
  if $N$ is reticulation-visible, then the \APDlong can be computed in polynomial time, in a bottom-up manner using the following lemma,
  contrasting the fact that computing~$\APDX$ is \#P-hard in general~\cite{van2025average}.
	
	\begin{lemma}\label{lem:visible}
    Let~$uv$ be an edge of a (not necessarily binary) reticulation-visible network~$N$.    
    Let~$R_v$ be the set of reticulations that can be reached via tree-paths from $v$ and,
    for each $r\in R_v$, let $Q_r$ be the set of edges on $v$-$r$-tree-paths whose head is~$r$.
    Then,
    \[
      \edgeproba{uv}{X_N} =
      \begin{cases}
        p_I(uv) & \text{if $v$ is visible in $N$}\\
        1 - \prod\limits_{r\in R_v} \left(1 - \sum\limits_{qr\in Q_r}p_I(qr)\right) & \text{otherwise}
      \end{cases}
    \]
	\end{lemma}
	\begin{proof}
    If $v$ is visible in $N$, then there is some leaf~$\ell$ such that all root-$\ell$ paths in $N$ contain~$v$.
    Since each switching contains a root-$\ell$-path, it also contains a $v$-$\ell$-path,
    so the probability of picking a switching in which~$uv$ is on a root-leaf-path is the
    probability of picking a switching that \emph{contains}~$uv$.
    This probability is $p_I(uv)$.
    Suppose now that $v$ is not visible in $N$.
    Since $N$ is reticulation-visible, $v$ is a tree-vertex and has children~$w_1,\dots,w_t$.
    Then,\\
    \phantom{xxxxxx} \begin{minipage}{.8\textwidth}
    \begin{enumerate}
        \item[(a)] $uv$ is \emph{not} on a root-leaf-path in a switching~$\sigma$
        \item[$\iff$ (b)] none of the edges~$vw_i$ are on a root-leaf-path in $\sigma$
        \item[$\iff$ (c)] none of the edges~$qr\in \bigcup_{r\in R_v} Q_r$ are in $\sigma$.
    \end{enumerate}
    \end{minipage}\\
    %
    Note that, for (c), being present in $\sigma$ is sufficient for~$qr$ to have a path to a leaf in $\sigma$
    since each reticulation is visible from a leaf~$\ell$ in $N$ and,
    thus, has a path to~$\ell$ in $\sigma$.
    Since the incoming edge of each reticulation is chosen independently,
    the probability of (c)~is the product over all~$r\in R_v$ of the probability that, among all edges incoming to~$r$, no edge of $Q_r$ is in $\sigma$.
    The probability for this is $1-\sum_{qr\in Q_r}\edgeproba{qr}{X_N} = 1-\sum_{qr\in Q_r}p_I(qr)$ by the first case.
	\end{proof}

  To realize the bottom-up computation, we can consider the trees resulting from the removal of all reticulations from~$N$.
  For each tree-vertex~$v$ in such a component, the term~$\edgeproba{uv}{X_N}$ can then be computed by recursively combining the values of the children.
  A dynamic program can do this in time and space linear in the size of the component.
  In total, this is linear in the size (number of vertices and edges) of $N$.
	
	\begin{theorem}\label{thm:visible}
		\APDX can be computed in linear time on reticulation-visible networks.
	\end{theorem}

	\subsection{Parameterization by invisibility}
	We now generalize \Cref{thm:visible} to networks with a bounded number of invisible reticulations.
  We generalize this to networks containing only a constant number of reticulations that are not visible, which we call ``invisible reticulations''.
  Our algorithm has two parts:
  First, we show how to compute the \APDlong for $N$
  by considering all partial switchings of invisible reticulations and applying \cref{thm:visible} to each of them.
  Then, we show how to remove a lowest biconnected component, given its \APDlong.
  In this way, we can handle biconnected components independently and,
  thus, reduce the number of reticulations that need to be considered simultaneously.

  \begin{lemma}\label{lem:bcc}
    Let~$R_I$ be the set of invisible reticulations of $N$.
    Then, $\APD(X_N)$ can be computed in $\prod_{r\in R_I}\indeg[N](r) \cdot O(m)$~time.
  \end{lemma}
  \begin{proof}
    Let~$R$ be the set of reticulations of $N$ and note that $R_I\subseteq R$.
    Further, let~$\swis_I(N)$ refer to the set of (partial) switchings for $R_I$ of $N$.
    Then, by \cref{obs:combine-reti-swis}, 
    $$\APD(X_N)
      = \smashsum{\sigma\in\swis(N)} P(\sigma) \PD[\hat{\sigma}](X_N)
      = \smashsum[l]{\sigma_I\in\swis_I(N)} \left(\smashsum[r]{\sigma\in\swis(\sigma_I)} P(\sigma) \PD[\hat{\sigma}](X_N)\right)
      \stackrel{\text{Obs.~\ref{obs:combine-reti-swis}}}{=}
    \smashsum{\sigma_I\in\swis_I(N)} P(\sigma_I)\cdot\APD[\hat{\sigma}_I](X_N).$$
    Since each~$\sigma_I\in\swis_I(N)$ is reticulation-visible,
    $\APD[\sigma_I](X_N)$ can be computed in linear time on~$\sigma_I$ by \cref{thm:visible}.
    Since $|\swis_I(N)|=\prod_{r\in R_I}\indeg[N](r)$, the claim follows.
  \end{proof}

  It remains to show that we can safely remove biconnected components, given their \APDlong.
  Indeed, we show that we can just remove a lowest biconnected component and treat it
  separately since the \APDlong of $N$ is just the sum of the two separate parts.

\newcommand\lemblob[1]{
	\begin{lemma}
		#1
    Let $B$ be a biconnected component of $N$ whose root~$\rho_B$ is minimal with respect to~$\belo$\todos{Term not defined}.
    Let $N'$ be the result of
    removing all vertices of $B$ except~$\rho_B$ and
    labeling~$\rho_B$ with any label in~$X_B$.
    Then, $\APD[N](X_N) = \APD[B](X_B) + \APD[N'](X_{N'})$.
  \end{lemma}
}
\lemblob{\label{lem:blob}}
  \begin{proof}
    We remark that, in each~$\sigma_B\in\swis(B)$, the root~$\rho_B$ has a path to a leaf, so $\haspath[\sigma]{e}{X_N}$
    for each~$e\in E(N')$ and $\sigma\in\swis(\sigma_B)$ does not depend on $\sigma_B$.
    Thus,
    \begin{align*}
      \APD(X_N)
      & \stackrel{\text{Obs.~\ref{obs:combine-reti-swis}}}{=}
        \smashsum{\sigma_B\in\swis(B)} \APD[\hat{\sigma}_B](X_N)\\
      & \stackrel{\eqref{eq:APD2025}}{=}
        \smashsum[l]{\sigma_B\in\swis(B)}\;
          \sum_{e \in E(N)}\w(e)
          \;\smashsum{\sigma\in\swis(\sigma_B)}\; P(\sigma) \cdot \haspathzo[\sigma]{e}{X_N}\\
      & =
        \smashsum[l]{\sigma_B\in\swis(B)} \left(
          \smashsum[r]{e \in E(B)}\w(e)
            \;\smashsum{\sigma\in\swis(\sigma_B)}\; P(\sigma) \cdot \haspathzo[\sigma]{e}{X_N} + 
              \smashsum{e \in E(N')}\w(e)
                \;\smashsum{\sigma\in\swis(\sigma_B)}\; P(\sigma) \cdot \haspathzo[\sigma]{e}{X_N}\right)\\
      & =
        \APD[B](X_B) +
        \sum_{\sigma_B\in\swis(B)}\;
          \sum_{e \in E(N')}\w(e)\cdot \;\;\smashsum{\sigma\in\swis(\sigma_B)}\; P(\sigma) \cdot \haspathzo[\sigma]{e}{X_N}
      \intertext{using $P(\sigma)=\prod_{e\in E(\sigma)}p_I(e)$,}
      & =
        \APD[B](X_B) +
        \smashsum[l]{\sigma_B\in\swis(B)}\;
          \smashsum[r]{e \in E(N')}
            \w(e) \;\smashsum{\sigma\in\swis(\sigma_B)}\;
              \haspathzo[\sigma]{e}{X_N} \cdot P(\sigma_B)\cdot  \;\smashprod{a\in E(\sigma)\cap E(N')}\; p_I(a)\\
      & =
        \APD[B](X_B) +
        \smashsum{\sigma_B\in\swis(B)}\;
          P(\sigma_B) 
          \;\smashsum{e \in E(N')}\;
            \w(e) \;\smashsum{\sigma\in\swis(\sigma_B)}\;
              \haspathzo[\sigma]{e}{X_N}\cdot
              \;\smashprod{a\in E(\sigma)\cap E(N')}p_I(a)\\
      \intertext{since $\haspath[\sigma]{e}{X_N}\iff\haspath[\sigma']{e}{X_N}$ for each $e\in E(N')$ and all $\sigma,\sigma'$ that agree on $E(N')$,}
      & =
        \APD[B](X_B) +
        \smashsum{\sigma_B\in\swis(B)}\;
          P(\sigma_B) 
          \;\smashsum{e \in E(N')}\;
            \w(e) \;\smashsum{\sigma'\in\swis(N')}\;
              \haspathzo[\sigma']{e}{X_N} \cdot
              \smashprod{a\in E(\sigma')}p_I(a)\\
      & \stackrel{\eqref{eq:APD2025}}{=}
        \APD[B](X_B) +
          \underbrace{\smashsum[r]{\sigma_B\in\swis(B)}\;  P(\sigma_B)}_{=1} \cdot \APD[N'](X_{N'})\\
      & = \APD[B](X_B) +
      \APD[N'](X_{N'})\qedhere
    \end{align*}
  \end{proof}
  \Cref{lem:bcc} and \cref{lem:blob} imply the following theorem.

	\begin{theorem}\label{thm:ivis-lvl}
    For each biconnected component~$C$ of a network~$N$,
    let~$R_C$ be the set of invisible reticulations in $C$ and
    let~$m$ denote the number of edges in~$N$.
    Then,
    \APDX can be computed in 
    $\Oh\big(m\cdot\max_C \prod_{r\in R_C}\indeg[N](r)\big)$~time.
	\end{theorem}

    \Cref{thm:ivis-lvl} implies that computing~$\APDX$ is fixed-parameter tractable (FPT) with respect to the number of invisible reticulations per biconnected component, when the maximum in-degree is a constant; in particular on binary networks.

	\section{Discussion}\label{sec:discussion}\looseness=-1
	In this paper, we extended the state of the art on an important problem for guiding conservation priorities when financial and temporal resources are limited.
	Specifically, we developed algorithms for computing the \APDlong (APD)
	of a given set of taxa whose evolutionary history includes reticulate events such as hybrid speciation, lateral gene transfer, and recombination.
	Our main contributions are bounds on the theoretical time complexity of this problem using the recently introduced graph parameter \emph{scanwidth}.
	Additionally, we show that it can be solved in linear time on so-called reticulation-visible networks and
    we generalize this result to networks that are close to being reticulation-visible by
    introducing the \emph{invisibility level}, the maximum number of invisible reticulations in any biconnected component.
    
	While our results significantly advance the theoretical understanding of APD computation in the presence of reticulate evolution,
    the practical implementation of these algorithms, along with their empirical evaluation and comparison on real and simulated datasets,
    remains an important direction for future work.
	Further, several theoretical open questions with potentially significant practical impact remain after our work.
	First, can our scanwidth-based algorithm be improved to the smaller node-scanwidth (see~\cite{berry2020scanning,Bruchhold24,holtgrefeThesis})?
	If so, would this modification lead to improved practical running times on empirical datasets? 
  Another promising research direction is to develop a ``visibility-aware'' version of scanwidth, as we did for the level,
  and see if APD can still be computed efficiently on networks with small ``invisibility scanwidth''.
  How exactly such a notion would be defined is unclear so far.
	A further avenue of research is to restrict the space of phylogenetic networks by focusing on specific graph classes.
	For instance, can the APD score of a given set of taxa be computed in polynomial time on orchard networks?
	In addition, several open questions arise for the maximization variant of our problem on special graph classes,
  with perhaps the most practically relevant being:
	Can \MaxAPD be solved, in polynomial time, on tree-child networks, and does this problem become any easier on ultrametric networks?

		
		

\bibliography{ref}



\appendixproofs

\end{document}